\documentclass[12pt]{amsart}
\usepackage{amsmath}
\usepackage{amssymb}
\usepackage{amsthm}
\usepackage[onehalfspacing]{setspace}

\usepackage[foot]{amsaddr}

\usepackage{natbib,verbatim,subfigure}
\usepackage{tikz}

\usetikzlibrary{decorations,calc,intersections,shapes,spy}

\usepackage[left=1.5in,top=1.5in,right=1.5in,bottom=1.5in]{geometry}

\def\Prod{\Pi}
\def\U{\mathcal{U}}

\def\R{\mathcal{R}}
\def\Rmon{\mathcal{R}^{\mbox{mon}}}

\def\Na{\mathbf{N}}

\def\W{\Omega}

\def\Li{\text{Li}}
\def\Ls{\text{Ls}}

\def\ep{\varepsilon}

\def\os{\varnothing}

\def\Re{\mathbf{R}}
\def\Qe{\mathbf{Q}}
\def\Na{\mathbf{N}}

\def\Rbmon{\overline{\mathcal{R}}^{\mbox{mon}}}
\def\Rbls{\overline{\mathcal{R}}^{\mbox{ls}}}

\def\calv{\mathcal{V}}

\newcommand{\df}[1]{\textit{#1}}
\newcommand{\norm}[1]{\| #1 \|}

\newcommand{\abs}[1]{ \left | #1 \right | }

\newdimen\slantmathcorr
\def\oversl#1{
\setbox0=\hbox{$#1$}
\slantmathcorr=\wd0
\hskip 0.2\slantmathcorr \overline{\hbox to 0.8\wd0{%
\vphantom{\hbox{$#1$}}}}
\hskip-\wd0\hbox{$#1$}
}

\def\undersl#1{
\setbox0=\hbox{$#1$}
\slantmathcorr=\wd0
\underline{\hbox to 0.8\wd0{%
\vphantom{\hbox{$#1$}}}}
\hskip-0.8\wd0\hbox{$#1$}
}

\def\ul{\underline}

\theoremstyle{plain}
\newtheorem{theorem}{Theorem}

\newtheorem{proposition}[theorem]{Proposition}
\newtheorem{lemma}[theorem]{Lemma}

\theoremstyle{definition}
\newtheorem{definition}[theorem]{Definition}
\newtheorem{example}[theorem]{Example}

\theoremstyle{remark}

\newtheorem*{claim*}{Claim}

\newtheoremstyle{named}{}{}{\itshape}{}{\bfseries}{.}{.5em}{#1\thmnote{
    #3}}
\theoremstyle{named}

\newtheoremstyle{named2}{}{}{\itshape}{}{\bfseries}{:}{.5em}{#1\thmnote{
    #3}}
\theoremstyle{named2}

\begin{document}
\title[Preference Identification]{Preference Identification}

\author[Chambers]{Christopher P. Chambers}
\author[Echenique]{Federico Echenique}
\author[Lambert]{Nicolas S. Lambert}

\address[Chambers]{Department of Economics, Georgetown University}
\address[Echenique]{Division of the Humanities and Social Sciences,
  California Institute of Technology}
\address[Lambert]{Graduate School of Business, Stanford University}


\thanks{We thank seminar audiences at Boston College, Boston University, Brown, Essex, LSE, Maryland, Johns Hopkins, George Mason, Ohio State, UCL, and audiences of conferences and workshops at Paris School of Economics, UPenn, Warwick and University of York. Echenique thanks the National Science Foundation for its support through the grants SES 1558757 and  CNS 1518941.
Lambert gratefully acknowledges the financial support and hospitality of Microsoft Research New York and the Yale University Cowles Foundation.}
\date{First version: April 10, 2017; This version: January 12, 2018.}
\begin{abstract}
An experimenter seeks to learn a subject's preference relation.  The experimenter produces pairs of alternatives. For each pair, the subject is asked to choose. We argue that, in general, large but finite data do not give close approximations of the subject's preference, even when the limiting (countably infinite) data are enough to infer the preference perfectly. We provide sufficient conditions on the set of alternatives, preferences, and sequences of pairs so that the observation of finitely many choices allows the experimenter to learn the subject's preference with arbitrary precision. While preferences can be identified under our sufficient conditions, we show that it is harder to identify utility functions. We illustrate our results with several examples, including consumer choice, expected utility, and preferences in the  Anscombe-Aumann  model.
\end{abstract}
\maketitle

\section{Introduction}

Consider a subject who forms a preference over the objects, or alternatives, in some set $X$. The subject participates in an experiment in which he is presented with a sequence of pairs of alternatives. For each pair, the subject is asked to choose one of the two alternatives offered.
What can an experimenter learn about the subject's preference from observing these binary comparisons?
Suppose that, after every observation, the experimenter computes an estimate of the subject's preference consistent with the data observed up that point: the experimenter chooses a preference rationalizing the choices made by the subject. Is the estimate a good approximation of the subject's underlying preference, for a large but finite experiment?

In this paper, we investigate the asymptotic behavior of preference estimates from finite experiments. We ask if one can fully identify the preference of a subject at the limit with finite data. It is a question of preference identification, in the classical sense of the term.\footnote{
Standard decision-theoretic language reserves the term \emph{identified} for a relation between preference and utility. In that context, a model is identified if every preference relation is represented by a unique (up to some class of transformations) set of parameters.  Thus, identification in this sense requires the knowledge of an entire preference relation. In this paper, we do not assume knowledge of the entire preference relation. Instead, we ask if one can learn the entire preference relation with a possibly large, but nonetheless finite data set. We discuss the identification of utility functions in Section~\ref {sec:utilityv1}, and the relation to decision theory in Section~\ref{sec:utilitydiscussion}.}

To illustrate the key issues, consider the following example.
Let $X \subseteq \Re^n$ represent a set of consumption bundles. The subject has a preference, denoted by $\succeq^*$, over the elements of $X$.  Over time, the subject is asked to choose an alternative from sets $B_k=\{x_k,y_k\}$, where $k$ is the time index. Together, the sets $B_1, B_2, \dots, B_k$ form a \emph{finite experiment.}
The experimenter observes the subject's choice of bundle for every pair. Assume the choice is consistent with the subject's preference, so that if $x$ is chosen over $y$, then $x \succeq^* y$.
Note that we can only, at best, infer the preference of the subject on the set $B \equiv \cup_{k=1}^\infty B_k$. Thus, if the subject's preference behaves very differently outside of the set $B$, there is no hope to obtain a fine approximation of the subject's preference over the entire set $X$. Two natural conditions emerge. First, we require that $\succeq^*$ be continuous, so one can hope to approximate the preference from finitely many samples. Second, we require that the set $B$ is dense in $X$, so that the observations are sufficiently spread out. And indeed, we show that, under these conditions, if one can observe the preference of the subject over the whole set $B$, then one can infer precisely $\succeq^*$ on $X$.

The two conditions, continuity of $\succeq^*$ and denseness of $B$, are, however, {\em not} enough to provide good approximations of $\succeq^*$ from {\em finitely} many observations. Knowledge of the preference over the infinite set $B$ allows the experimenter to exploit the continuity assumption on the subject's preference. With finite data, continuity does not have enough bite.
To illustrate, take $X=[0,1]$, and suppose that the preference of the subject over $X$ is captured by the binary relation $\geq$ (greater numbers are always preferred). Consider the countable set of objects $B=\Qe \cap (0,1)$, and $B_1, B_2, \dots$ an enumeration of pairs of objects of $B$.
Then, any continuous preference that agrees with $\geq$ on $\Qe$ has 1 weakly preferred to 0. However, for any $n$, one can find a preference $\succeq_n$ that rationalizes the choices of the subject over $B_1,\dots,B_n$, and yet that ranks 0 strictly above 1. 

In fact, one can come up with an even more startling example: we show that no matter the subject's preference, the experimenter may end up inferring that the subject is indifferent among all alternatives (see Section~\ref{sec:motivation}). And yet, as in the example just described, she would be able to infer the subject's preference perfectly, had she access to the subject's preference over the infinite set $B$ all at once. The example exhibits a kind of discontinuity. With infinite data in the form of $B$, we must conclude that $x\succ y$, but any finite data cannot rule out that $y\succeq x$.

These examples illustrate the dangers of data-driven estimation. Non-parametric  estimation with finite data can behave very differently from estimation with infinite (even countable) data. To derive meaningful estimates, one must construct a theory that disciplines the preferences, and lays down the proper conditions for convergence of preference estimates.

This paper includes three sets of results.

Our first and foremost results concern non-parametric estimation. We offer fairly general conditions so that observing sufficiently many binary choices allows one to approximate the subject's preference arbitrarily closely with \emph{any preference that rationalizes the finite data}.

We provide two notions of rationalization, a weak and a strong one.
Under strong rationalization, a rationalizing preference must reflect choices perfectly. So if one alternative is chosen over another, the preference must rank the first strictly above the second. Under weak rationalization, the first alternative must only be ranked at least as good as the second. Weak alternatives reflect the phenomenon of {\em partial observability} \citep{CES} whereby one cannot infer anything from a choice that was not made.

Under both notions of rationalization, it is necessary to add structure on the environment and on the rationalizing preferences so as to avoid the negative results of the example above. Importantly, we need a notion of {\em objective} rationality expressed by the monotonicity of preferences.
We postulate an exogenous partial ordering of the set of alternatives---for example, standard vector dominance when the set of alternatives represents consumption bundles, or stochastic dominance when it is the set of lotteries over monetary amounts---and we require that the subject's preference is monotonic with respect to that exogenous order.

With this structure, finite-experiment rationalizable preferences converge to the subject's underlying preference, under conditions that are consistent with many applications in decision theory, and with experimental implementations of decision theoretic models. Stronger conditions are needed to obtain the result for weak rationalization---conditions that hold for preferences over Euclidean spaces, but rule out some common applications in decision theory---yet it is remarkable that convergence is at all attained under weak rationalization. By remaining agnostic about choices that were not made, we are inferring a lot less about the subject's preferences under the assumption of weak rationalization than under strong.

Our results on preference identification are relevant to a wide range of contexts. For concreteness, we illustrate their application to the special case of preferences over lotteries, dated rewards, consumption bundles, and Anscombe-Aumann acts \citep{anscombe}. In all these cases, there is a natural objective  partial order and monotonicity is a sensible assumption. There are other environments in which one cannot reasonably impose any kind of monotonicity. For instance, in the literature on discrete allocation (such as in \citet{hylland1979efficient} or in the recent literature on school choice, such as in \citet{sethuraman2011}) in which agents are assumed to choose among lotteries over finitely many heterogeneous objects. Monotonicity would require that all agents agree on a ranking of the discrete objects that are being allocated, an unreasonable requirement. 

Our second set of results concerns the identification of utility functions. Given a utility representation for the agent's preference, we show that it is possible to carefully select finite-data utility rationalizations so as to approximate the subject's utility arbitrarily closely. This result also rests on monotonicity assumptions. However, there is a clear difference between estimating preferences and utilities. While any preference estimate converges to the true underlying preference, for utilities we only know that a certain selection converges. This observation is especially relevant when estimating utilities of a particular functional form. There is no guarantee that such utility estimates have the correct asymptotic behavior; one can only say that the preferences that these utilities represent do.

Our third and final results concern the identification of preferences with infinite but countable data. We show that, when the experimenter has access to the preference of the subject over all alternatives of a countable set, then it is possible to recover perfectly the subject's preference over the entire set of alternatives $X$ under much weaker conditions than above. We further demonstrate that, under such conditions, the experimenter can, in theory, obtain the subject's preference directly from the observation of a single choice of the subject when the subject is asked to select an object among a large, infinite set.

The remainder of the paper proceeds as follows. After reviewing the literature, we describe the model in Section~\ref{sec:model}. We present our main results for important cases of collections of alternatives in Section~\ref{sec:results} and discuss these results in Section~\ref{sec:discussion}. In Section~\ref{sec:mainresults}, we present our main results for general collections of alternatives. In Section~\ref{sec:utility}, we study the relation between preference and utility, and provide conditions under which the identification of a preference makes it possible to identify a utility, and conversely. In Section~\ref{sec:expectedutility}, we further show that when the set of possible utilities is compact, one can obtain a strong form of identification, which dispenses with the postulate of existence of a data-generating preference. We discuss the question of preference identification with infinite but countable data in Section~\ref{sec:infinitedata}. Finally, in Section~\ref{sec:discusspreference}, we offer interpretations on the meaning of a data-generating preference. We relegate the proofs and more technical results to the appendices (some of these results may be of independent interest).

\subsection*{Literature Review}
Experimentalists and decision theorists have an obvious interest in preference  estimation, but we are not aware of any study of the behavior of preference estimates from finite experiments. The long tradition of revealed preference theory from finite data that goes back to \cite{afriat} is focused on testing, not estimation.
The closest work to ours is \cite{mascolell78}, who works with finitely many observations from a demand function over a finite number of goods. Mas-Colell assumes  a rational demand function that satisfies a boundary condition and is ``income Lipschitzian.'' He assumes
a sufficiently rich sequence of observations, taken from an increasing sequence of budgets. He then shows that the sequence of rationalizing preferences, each rationalizing a finite (but increasing) set of observations, converges to the unique preference that rationalizes the demand function.

There are many differences between Mas-Colell's exercise and ours, even if one restricts attention to choice over bundles of finitely-many, divisible, consumption goods. In particular, the difference in model primitives---demand instead of binary comparisons---is crucial. One cannot generally use choice from linear budgets to recreate any given binary comparison. Moreover, there is no property analogous to the boundary and Lipschitz continuity of demand in our framework. Indeed, as shown in \cite{mascolell77}, by means of an example due to Lloyd Shapley, without these properties, preferences are not identified from demand.\footnote{Shapley's example also appears in \cite{raderbook}. The example poses no problem for identification in our framework of binary comparisons. It generates non-identification of demand because two preferences have the same marginal rate of substitution at the sampled points. With binary comparisons, the differences between two such preferences are detected.} In Mas-Colell's paper, weak and strong rationalizability coincide, as he works with demand functions. In this paper, we are particularly interested in partial observability.

Also working with demand functions, the recent stream of literature by \cite{reny2015characterization}, \cite{kubler2015identification} and  \cite{polemarchakis2016identification}  provide results on the limiting behavior of finite-data utility rationalizations. These papers focus on the convergence of certain utility constructions that rationalize finite demand data. In contrast, our main results are about the convergence of (any) rationalizing preferences. There are also important differences between the primitives assumed in our paper and the demand functions assumed in these papers.

The recent paper by \cite{gorno} also looks at the identification of preferences from abstract choice behavior. A clear difference between Gorno's exercise and ours is that we consider the limiting behavior of large, but finite, experiments. His paper deals with preference identification from a given choice behavior on a fixed choice set. While the two papers are concerned with related questions, the exercises are quite different and the results are not related.

Finally, on the technical level, we use the topology on preferences introduced by \cite{HILDENBRAND1970161} and \cite{kannai1970continuity}, building on the work of \cite{debreu}. In our study of the mapping from utility to preference, we borrow ideas from \cite{mas1974continuous} and \cite{border1994dynamic}. In particular, the proof of the continuity of the ``certainty equivalent'' representation is analogous to Mas-Colell's, and we take the notion of local strictness from Border and Segal, as well as their continuity result.

\section{Model}\label{sec:model}

In this section, we introduce the definitions and conventions used throughout the paper, and present our main model.  Our focus in this section is on two classical environments; namely, consumption space and Anscombe-Aumann acts.

\subsection{Basic definitions and notational conventions}

Let $X_i$ be a set partially ordered by $\geq_i$, for $i=1,\ldots,n$.\footnote{A partial order is reflexive, transitive, and anti-symmetric, while a strict partial order is irreflexive, transitive, and asymmetric.}
If $x,y\in \Prod_{i=1}^n X_i$, then $x\geq y$ means that $x_i\geq_i y_i$ for $i=1,\ldots, n$; and $x> y$ that $x\geq y$ and $x\neq y$. We write $x \gg y$ when $x_i>_i y_i$ for $i=1,\ldots, n$. The order $\geq$ on $\Prod_{i=1}^n X_i$ is called the \df{product order}.\footnote{Then $\geq$ is a partial order, and each of $>$ and $\gg$ are strict partial orders.}
The \df{interval} $[a,b]$ in $\Prod_{i=1}^n X_i$ denotes the set $\{z\in\Prod_{i=1}^n X_i:b\geq z\geq a\}$. An \df{open interval} $(a,b)$ denotes the set $\{z\in \Prod_{i=1}^n X_i:b\gg z\gg a\}$.
When $X_i=\Re$, and $\geq_i$ is the usual order on the real numbers, the above definitions constitute the familiar ordering on $\Re^n$, as well as the usual notions of intervals and open intervals.

If $A\subseteq \Re$ is a Borel set, we write $\Delta(A)$ for the set of all Borel probability measures on $A$. We endow $\Delta(A)$ with the weak* topology.
For $x,y\in\Delta(A)$, we write $x\geq_{FOSD} y$ when $x$ is larger than on $y$ in the sense of \df{first order stochastic dominance} (meaning that $\int_A fdx\geq \int_A fdy$ for all monotone increasing, continuous and bounded functions $f$ on $A$).  When $\Omega$ is a finite set,
we shall use the above definitions to order $\Delta(A)^\Omega$ by the product order defined from ordering $\Delta(A)$ by first-order stochastic dominance.

For an integer $n$, $[n]$ denotes the set $\{1,\ldots,n \}$. So $\Delta([n]) = \{x\in\Re^n_+:\sum_{i=1}^nx_i=1\}$ is the \df{simplex} in $\Re^n$.

\subsubsection{Preference relations}
Let $X$ be a set.  Given a binary relation $R\subseteq X\times X$, we write $x \mathrel{R} y$ when $(x,y)\in R$. And we say that a function $u:X\rightarrow\Re$ \df{represents} $R$ if $x \mathrel{R} y$ iff $u(x)\geq u(y)$.
A \df{preference}, or \df{preference relation}, is a weak order; \emph{i.e.} a binary relation over $X$ which is complete and transitive..

For a partial order $\geq$ on $X$, a preference $\succeq$ on $X$ is \df{weakly monotone (with respect to $\geq$)} if $x \geq y$ implies that $x \succeq y$. For a strict partial order $>$ on $X$, a preference $\succeq$ on $X$ is \df{strictly monotone (with respect to $>$)} if $x > y$ implies that $x \succ y$.\footnote{The strict part of a partial order is a strict partial order, but sometimes we are interested in other partial orders.  For example, $\gg$ is not the strict part of $\geq$, but we use it later.}

A binary relation $R$ is \df{continuous} if $R\subseteq X\times X$ is closed (see, for example, \citealp{bergstrom1976}).

The set of continuous binary relations over $X$, when $X$ is a topological space, is endowed with the \df{topology of closed convergence}, we provide a definition in Section~\ref{sec:convpref}. It is the natural topology for our purposes because it is the weakest topology for which optimal behavior is continuous.

Under the assumptions of our paper, the topology of closed convergence is the smallest topology for which the sets
\[
	\{ (x,y,\succeq) : x\succ y\}
\]
are open (see \citealp{kannai1970continuity},  Theorem 3.1). So it is the weakest topology with the following property:
Suppose that a subject with preferences $\succeq$ chooses $x$ over $y$ ($x\succ y$), and that another subject with sufficiently close preferences $\succeq'$ face a choice between $x'$ and $y'$, where $(x',y')$ is sufficiently close to $(x,y)$, then the second subject must choose $x'$ over $y'$ ($x'\succ' y'$).  In other words, the optimal behavior according to $\succeq'$ approximates the optimal behavior according to $\succeq$. In this sense, the topology that we impose on preferences is natural in any investigation of optimal choice.

\subsubsection{Choice functions}

A pair  $(\Sigma,c)$ is a \df{choice function} if  $\Sigma\subseteq 2^X\setminus\{\os\}$ is a collection of nonempty subsets of $X$, and $c:\Sigma\to 2^X$ with $\os\neq c(A)\subseteq A$ for all $A\in\Sigma$. When $\Sigma$, the domain of $c$, is implied, we refer to $c$ as a choice function.

A choice function $(\Sigma,c)$ is \df{generated} by a preference relation $\succeq$ if \begin{equation*}
	c(A)= \{x\in A : x\succeq y \text{ for all } y\in B \},
\end{equation*}
for all $A \in \Sigma$.

The notation $(\Sigma,c_{\succeq})$ means that the choice function $(\Sigma,c_{\succeq})$ is generated by the preference relation $\succeq$ on $X$.

\subsection{The model.}

There is an experimenter (a female) and a subject (a male). The subject chooses among alternatives in a set $X$ of possible alternatives. The subjects' choices are guided by a preference $\succeq^*$ over $X$, which we refer to as data-generating preference. The experimenter seeks to infer $\succeq^*$ from the subject's choices in a finite experiment.

In a finite experiment, the subject is presented with finitely many unordered pairs of alternatives $B_k=\{x_k,y_k\}$  in $X$. For every pair $B_k$, the subject is asked to choose one of the two alternatives: $x_k$ or $y_k$.

A \df{sequence of experiments} is a collection  $\Sigma_{\infty} = \{B_i\}_{i \in \Na}$ of pairs of possible choices presented to the subject. Let $\Sigma_k=\{B_1, \dots, B_k\}$, and let $B = \cup_{k=1}^\infty B_k$ be the set of all alternatives that are used over all the experiments.

We make two assumptions on $\Sigma_\infty$. The first is that $B$ is dense in $X$. The second is that, for any $x,y\in B$ there is $k$ for which $B_k=\{x,y\}$. The first assumption is obviously needed to obtain any general identification result (see Section~\ref{sec:denseness}). The second assumption means that the experimenter is able to elicit the subject's choices over all pairs used in her experiment.\footnote{If there is a countable dense $A\subseteq X$, then one can always construct such a sequence of experiments via a standard diagonalization argument.}

\subsubsection{The data and its rationalizations.}

For each $k$, the subject's preference
$\succeq^*$ generates a choice function $(\Sigma_k,c_{\succeq^*})$. Thus the choice behavior observed by the experimenter is always consistent with $(\Sigma_k,c_{\succeq^*})$. We term $(\Sigma_k,c_{\succeq^*})$ the \df{choice function of order $k$ generated by $\succeq^*$}, and we term the choice function $(\Sigma_\infty,c_{\succeq^*})$ the \df{choice sequence of order $k$ generated by $\succeq^*$}.

Sometimes, we may not be able to observe the subject's entire choice function.  In the spirit of \citet{afriat}, we want to allow for the possibility that the subject may in principle be willing to choose $x$, but does not actually choose it.  In the language of \citet{CES}, we want to study the concept of \df{partial observability}.  To this end, a general choice function $(\Sigma_{\infty},c)$ is termed a \df{choice sequence} and this induces, for every $k$, a choice function on $\Sigma_k$.

For a choice function $c$ and a preference $\succeq^*$, we use the notation $c\sqsubseteq c_{\succeq^*}$ to mean that for each budget $B_k$, $c(B_k) \subseteq c_{\succeq^*}(B_k)$; \emph{i.e.}, the observed choices from $B_k$ are optimal for $\succeq^*$.   

In the context of partial observability, the notion of rationalization needs to accommodate the fact that some preference maximal alternatives may not have actually been chosen.  The next concept captures such an accommodation; and is again in the spirit of \citet{afriat}.

A preference $\succeq_k$ \df{weakly rationalizes} $(\Sigma_k,c)$  if, for all $B_i\in \Sigma_k$,  $c(B_i)\subseteq c_{\succeq_k}(B_i)$. A preference $\succeq_k$ weakly rationalizes a choice sequence $(\Sigma_\infty,c)$ if it rationalizes the choice function of order $k$ $(\Sigma_k,c)$, for all $k \ge 1$.

The following concept is analogous to the notion of rationalization discussed in \citet{richter}, and is the appropriate notion when it is known that all potentially chosen alternatives are actually chosen.

A preference $\succeq_k$ \df{strongly rationalizes} $(\Sigma_k,c)$  if, for all $B_i\in \Sigma_k$,  $c(B_i)= c_{\succeq_k}(B_i)$. A preference $\succeq_k$ strongly rationalizes a choice sequence $(\Sigma_\infty,c)$ if it rationalizes the choice function of order $k$ $(\Sigma_k,c)$, for all $k \ge 1$.

\nopagebreak[1000000]
\section{Results}\label{sec:results}
\nopagebreak[1000000]
\subsection{Motivation}\label{sec:motivation}
\nopagebreak[1000000]
Many results on identification in economics presume access to rich information. In decision theory, the presumption is that one can observe enough of the subject's choices so as to effectively know  the subject's preference $\succeq^*$. In this section, we point to some problems with this assumption.

Let $\succeq^I=X\times X$ denote the degenerate preference relation that regards any two alternatives as indifferent.

\begin{proposition}\label{prop:indiff} Let $X=[a,b]\subseteq\Re^n$, where $a\ll b$. Let the subject's preference $\succeq^*$ be continuous. There exists a continuous preference $\succeq_k$ that strongly rationalizes the choice function of order $k$ generated by $\succeq^*$, and such that $\succeq_k \rightarrow \succeq^I$.
\end{proposition}
The proof of Proposition~\ref{prop:indiff} is relegated to Appendix \ref{app:proof:proofofpropindiff}.

Proposition~\ref{prop:indiff} means that, absent further conditions, the sequence of rationalizations can be very different from the preference $\succeq^*$ generating the subject's choices. It is possible to choose rationalizations that converge to full indifference among all alternatives, regardless of which $\succeq^*$ really generated the subject's choices. The objective of our paper is to show how such problems can be avoided.

Proposition~\ref{prop:indiff} suggests another distinction. There is ``infinite data'' in the form of the data-generating preference $\succeq^*$, such data is commonly assumed in decision theory; there is
finite data, in the form of $(\Sigma_k,c_\succeq^*)$; and,  ``limiting data,'' which would be $\succeq^*|_B$, {\em i.e.,} the preference $\succeq^*$ restricted to domain $B$. With limiting data one {\em would} be able to identify $\succeq^*$. Indeed, we show in Section~\ref{sec:limitingdata} that if $\succeq|_B =\succeq^*|_B$ then $\succeq=\succeq^*$.

Therefore, Proposition~\ref{prop:indiff} illustrates a sort of discontinuity. If one only had access to limiting data, there would be no problem. However, with arbitrarily large, but finite, data, preference rationalizations can be completely wrong.

\subsection{Weak rationalizations} \label{sec:weakrat-v1}
We now present a series of simple sufficient conditions ensuring convergence of preference rationalizations to the subject's preference. The results are discussed in Section~\ref{sec:discussion}.

Let $X=\Re^n_+$. Recall that the strict partial order $\gg$ on $\Re^n_+$ refers to the relation $x\gg y$ if for each $i$, $x_i > y_i$ (\emph{i.e.} the product of $>$); strict monotonicity refers to this relation.

\begin{theorem}\label{thm:weakrat-v1} Let the subject's preference $\succeq^*$ be continuous and strictly  monotone.  Suppose that $c \sqsubseteq c_{\succeq^*}$.  For each $k\in\Na$,
let $\succeq_k$ be a continuous and strictly monotone preference that weakly rationalizes $(\Sigma_k,c)$. Then, $\succeq_k \rightarrow \succeq^*$.
\end{theorem}
Theorem~\ref{thm:weakrat-v1} can be generalized, as $\succeq^*$ and $\succeq_k$ do not need to be transitive. They only need to be continuous, strictly monotone, and complete.

\subsection{Strong rationalization}\label{sec:SMv1}
Suppose that $X$ is either
\begin{enumerate}
\item $\Re^n_+$,
\item or $\Delta([a,b])^\W$ for a finite set $\W$ and $[a,b]\subseteq\Re$.
\end{enumerate}

Recall that we topologize $\Delta([a,b])$ with its weak* topology, and $\Delta([a,b])^\W$ with the product topology.  In the case of $\Re^n_+$, the relation $\geq$ refers to the product of $\geq$ on each of the coordinates, and on $\Delta([a,b])^\W$, it is the product of $\geq_{FOSD}$.

\begin{theorem}\label{thm:weaklymon-v1} Let the subject's preference $\succeq^*$ be weakly monotone. For each $k\in\Na$, let
 $\succeq_k$ be a continuous and weakly monotone preference that strongly rationalizes the choice function of order $k$ generated by $\succeq^*$. Then, $\succeq_k \rightarrow \succeq^*$.
\end{theorem}

\subsection{Utility functions}\label{sec:utilityv1}
Let $X$ be either of the sets in Section~\ref{sec:SMv1}.  In the case of $\Delta([a,b])^\W$, $>$ references the strict part of $\geq$ as defined previously.

Denote by $\mathcal{R}^{\mbox{mon}}$ the set of preferences that are strictly monotone and continuous, and by $\mathcal{U}$ the set of strictly increasing and continuous utility functions on $X$. The set $\mathcal{U}$ is endowed with the topology of uniform convergence on compacta.

Let $\Phi$ be the function that carries each utility function in $\mathcal{U}$ into the preference relation that it represents. So  $\Phi:\mathcal{U}\rightarrow\R^{\mbox{mon}}$ is such that $x \mathrel{\Phi(u)} y$ if and only if $u(x) \geq u(y)$.

We regard two utility functions as equivalent if they represent the same preference: if they are ordinally equivalent.
Define an equivalence relation $\simeq$ on $\mathcal{U}$ by $u\simeq v$ if there exists $\varphi:\Re\rightarrow\Re$ strictly increasing for which $u=\varphi \circ v$.  Then let $\mathcal{U}/\simeq$ denote the set of equivalence classes of $\mathcal{U}$ under $\simeq$ endowed with the quotient topology.
The function $\hat{\Phi}:\mathcal{U}/\simeq\rightarrow \mathcal{R}^{\mbox{mon}}$ maps an equivalence class into $\Phi(u)$, for any $u$ member of the equivalence class.

\begin{theorem}\label{thm:homeomorphism-v1} $\hat{\Phi}$ is a homeomorphism.
\end{theorem}

Theorem~\ref{thm:homeomorphism-v1} implies that a utility representation may be chosen from a finite-experiment rationalization so as to approximate a given utility representation for the preference generating the choices.

The following Proposition adds some structure to Theorem~\ref{thm:homeomorphism-v1}.  It claims a lower hemicontinuity result for $\Phi^{-1}$, in the sense that for \emph{any} utility representation of a strictly monotone preference, a convergent sequence of preferences possesses a convergent sequence of utility representations.

\begin{proposition}\label{prop:lowersemicontinuity}
Let the subject's preference $\succeq^*$ be strictly monotone and continuous. Let $\succeq_k$ be a continuous and strictly monotone preference that strongly rationalizes the choice function of order $k$ generated by $\succeq^*$. Then,
for any utility representation $u^*$ of $\succeq^*$, there exist utility representations $u_k$ of $\succeq_k$ such that $u_k\rightarrow u^*$.
 \end{proposition}

\subsection{Limiting data}
\label{sec:limitingdata}
Let $X$ be as in Section~\ref{sec:SMv1}, and let $B$ be the dense set of alternatives used over all experiments.

\begin{theorem}\label{thm:connected-v1}
Suppose that $\succeq$ and $\succeq^*$ are two continuous preference relations.
If $\succeq|_{B\times B} = \succeq^*|_{B\times B}$, then $\succeq = \succeq^*$.
\end{theorem}

As we discussed in Section~\ref{sec:motivation}, the case of limiting data serves to illustrate the difference between a sequence of finite experiment and the limit of a countably infinite dataset. Theorem~\ref{thm:connected-v1} means that one can obtain identification of $\succeq^*$ solely from the continuity assumption (we refer the reader to Section~\ref{sec:infinitedata} for more details).

\section{Discussion}\label{sec:discussion}
\subsection{Positive results and assumptions on $X$.}

The previous section assumes that $X$ is either $\Re_+^n$ or $\Delta([a,b])^\W$, for some finite set $\W$. The point of Theorems~\ref{thm:weakrat-v1} and~\ref{thm:weaklymon-v1} is to provide a positive result in response to the concerns raised in Section~\ref{sec:motivation}. The results cover some of the most widely used choice spaces in economics: $\Re^n_+$ is consumption space in demand theory, and $\Delta([a,b])^\W$ is a space of Anscombe-Aumann acts over monetary lotteries.

For the Anscombe-Aumann interpretation,
let $\Omega$ be a finite nonempty set of \emph{states of the world}, and interpret $[a,b]$ as a set of monetary payoffs. The elements of $\Delta([a,b])$ are lotteries of monetary payoffs. An \df{Anscombe-Aumann act} is a state-contingent monetary lottery; it maps elements from $\Omega$ to $\Delta([a,b])$. The set of alternatives $\Delta([a,b])^{\Omega}$ is then the set all Anscombe-Aumann acts.

The spaces  $\Re_+^n$ and $\Delta([a,b])^\W$ have in common that there is an objective notion of monotonicity that preferences can be made to conform to. Other spaces share this property. Section~\ref{sec:mainresults} includes our most general results.

\subsection{Identification of utility functions}\label{sec:utilitydiscussion}
Many results on identification in decision theory can be phrased in the following terms. There are subsets $\U'\subseteq\U$ and $\R'\subseteq\Rmon$, and an equivalence relation $\simeq'$ on $\U'$ such that $\Phi$ is a bijection from $\U'/\simeq'$ onto $\R'$. The idea is that, with data in the form of $\succeq^*$, one can uniquely ``back out'' an equivalence class from $U'$.

Our results suggest that this is not enough when data is finite. First, one needs to ensure that rationalizations obtained from finite data converge to the underlying $\succeq^*$. Second, the space of preferences and utilities have to be homeomorphic in order to be able to obtain a limiting utility function from a large, but finite, dataset on choices.

\subsection{Partial observability}

The distinction between weak and strong rationalizability is important. In fact, it is rather surprising that one can obtain a result such as Theorem~\ref{thm:weakrat-v1} for weak rationalizations.

A choice sequence generated by the subject's preference reflects both strict comparisons as well as indifferences. In practice, however, the experimenter may not be able to properly infer the indifference of the subject regarding two alternatives. The difficulty arises, for example, when the experimenter offers the subject his preferred alternative. In this case, the experimenter would typically require that the subject selects only one of the two alternatives presented to him. Such situations, in which the experimenter cannot commit to being able to see all potentially chosen elements, are referred to \emph{partial observability} \citep{CES}, in contrast to full observability in which the experimenter is able to elicit the subject's indifference between alternatives.  

Weak rationalizability expresses the idea that the experimenter is not willing to commit to interpreting observed choices as the \emph{only} potential choices made by the subject. For example, if the experimenter observes that the subject chooses $x$ when presented the pair $\{x,y\}$, she may not be willing to infer that $x \succ^* y$, as it may be that $x \sim^* y$ but the subject simply did not choose $y$. This notion of weak rationalization is used, for example, by \citet{afriat} in the context of consumer theory (for more details on this notion, see, for example, \citet{rptheory}).\footnote{Analogously, the hypothesis of full observability, related to what we call strong rationalization, is the notion employed by \citet{richter}.  A recent work showing how to obtain both types of conditions as a special case is \citet{nishimura}.}

Weak rationalizability is partially agnostic with respect to the status of unchosen alternatives, so it is surprising that one can ensure convergence of preference rationalizations to the preference that generated the choices.

\subsection{Monotone Preferences}

The problem exemplified by Proposition~\ref{prop:indiff} is that one cannot hope to obtain convergence to $\succeq^*$ if there is no discipline placed on the rationalizing preferences. In a sense, we need to constrain, or structure, the theory from which rationalizations are drawn. Our results show that a notion of objective monotonicity is enough to ensure that rationalizing preferences in the limit approach the subject's preferences.

As discussed in the introduction, and exemplified by Proposition~\ref{prop:indiff},
the continuity assumption on the subject's preference, and the assumption that the alternatives offered are in the limit dense, do not generally ensure convergence to the subject's preference. Proposition~\ref{prop:indiff} shows that the failure of convergence can be rather dramatic. We must impose structure on the subject's preference, and on the finite-experiment rationalizations.

Observe that the preferences $\succeq_k$ constructed in Proposition~\ref{prop:indiff} cannot be monotone. Suppose that $\succeq^*$ is a continuous preference relation, and suppose that $x\succ^* y$. In the construction in Proposition~\ref{prop:indiff} we obtain a sequence of rationalizations $\succeq_k$ such that in the limit $y$ is at least as good as $x$. This cannot happen if each rationalizing preference is weakly monotone: $x\succ^* y$ implies that $x'\succ^* y'$ for $(x',y')$ close enough to $(x,y)$. Thanks to the  interaction of the order and the topology on $\Re^n$
we can find a $k$ large enough such that there are $\{x'',y''\}\in \Sigma_k$ (meaning alternatives offered in the $k$th finite experiment) with $x'\geq x''$ and $y''\geq y'$, and where $(x'',y'')$ is also close to $(x,y)$. If $\succeq_k$ is monotone then we have $x'\succeq_k x''$ and $y''\succeq y'$. But if $\succeq_k$ strongly rationalizes the choices made at the $k$th experiment, then $x''\succ_k y''$. So we have to have  $x'\succ_k y'$ for any $(x',y')$ close enough to $(x,y)$.

\subsection{On the denseness of $B$.}
\label{sec:denseness}
We assume that $B$, the set of all alternatives used in a sequence of experiments, is dense. Our paper deals with fully nonparametric identification, so it seems impossible to obtain a general result without assuming denseness of $B$: imagine that the experimenter leaves an open set of alternatives outside of her experimental design. Then the subject's preferences over alternatives in that set would be very hard to gauge.

In practice, one can imagine restricting attention to smaller class of families for which one does not need to elicit choices over a set that is dense in $X$. For example for expected utility preferences over lotteries, or homothetic  preferences in $\Re^n$, one is only trying to infer a single indifference curve. So a smaller set of choices is enough: but even in that case one would need the set of alternatives in the limit to be dense in the smaller set of choices.

\subsection{Preference identification and utility identification.}
The theorems say, roughly speaking, that, if we assume data generated by a (well behaved) preference $\succeq^*$, then {\em any} ``finite sample rationalization'' $\succeq^k$ is guaranteed to converge to the generating preference. So estimates have the correct ``large sample'' properties. In particular, one may be interested in a specific theory of choice, such as max-min or  Choquet expected utility. If the subject's $\succeq^*$ is max-min, or Choquet, one  can choose rationalizing preferences to conform to the theory, and the limit will uniquely identify the subject's max-min, or Choquet, preference. But if one incorrectly uses rationalizing preferences outside of the theory, the asymptotic behavior will still correct the problem, and uniquely identify $\succeq^*$ in the limit.  The theorems also say that there are certain utility representations $u^k$ that will be correct asymptotically.

Note, however, what the theorems do not say. First, the estimates $\succeq^k$ are guaranteed to converge to the generating preferences $\succeq^*$, \emph{when the generating preference is known to exist.} If one simply estimates the preferences $\succeq^k$, these may fail to converge to a well-behaved preference. We present two examples to this effect in Section~\ref{sec:discusspreference}. That said, under certain conditions (that unfortunately are not satisfied in the Anscombe-Aumann setting), the ``size'' of the set of rationalizing preferences shrinks as $k$ growth; see Theorem~\ref{thm:diameter}.

Second,  our results do not say that one can choose $u^k$ arbitrarily. Any estimated rationalizing preference will converge to the preferences rationalizing the utility, but basing the estimation on utilities is more complicated because it is not clear that any utility representation of $\succeq^*$ will have the right limit, or even converge at all.

\section{General Results}\label{sec:mainresults}

In Section~\ref{sec:results}, we have presented our main results for some important special cases of the collection of alternatives $X$. In this section, we present our main results for the general case. We now assume that $X$ is a Polish and locally compact space, and provide conditions under which our convergence results continue to hold. The conditions we provide are the weakest we know. The section concludes with applications to this general case. 
Note that this section focuses on preference identification, general results for utility identification are given in Section~\ref{sec:utility}.

\subsection{Convergence of Preferences}\label{sec:convpref}

To speak about the approximation of the subject's preference, one must introduce a notion of convergence on the space of preferences.
We use \emph{closed convergence}, and endow the space of preference relations with the associated topology.
The use of closed convergence for preference relations was initiated by the work of \cite{kannai1970continuity} and \cite{HILDENBRAND1970161}, and has become standard since then.

One primary reason to adopt closed convergence is to capture the property that agents with similar preferences should have similar choice behavior---a property that is necessary to be able to learn the preference from finite data. Specifically, under the assumptions we use for most of our results, the topology of closed convergence is the smallest topology for which the sets
\[
	\{ (x,y,\succeq) : x\succ y\}
\]
are open (see \cite{kannai1970continuity}  Theorem 3.1). The desired continuity of choice behavior is expressed by the fact that sets of the form $\{ (x,y,\succeq) : x\succ y\}$ are open. The topology of closed convergence being the smallest topology with this property is a natural reason for adopting it.

The  following characterization of closed convergence for the context of preference relations will be used throughout the paper:
\begin{lemma}\label{lem:tcclim}
Let $\succeq_n$ be a sequence of preference relations, and let $\succeq$ be a preference relation. Then $\succeq_n \rightarrow \succeq$ in the topology of closed convergence if and only if, for all $x,y\in X$,
\begin{enumerate}
	\item $x\succeq y$ implies that for any neighborhood $V$ of $(x,y)$ in $X\times X$ there is $N$ such that for all $n\geq N$, $\succeq_n\cap V\neq \os$;
	\item if, for any neighborhood $V$ of $(x,y)$ in $X\times X$, and any $N$ there is $n\geq N$  with $\succeq_n\cap V\neq \os$, then $x\succeq y$.
\end{enumerate}
\end{lemma}

The following lemma plays an important role in the approximation results.
\begin{lemma}\label{lem:closedconvergencecompact}
	The set of all continuous binary relations on $X$, endowed with the topology of closed convergence, is a compact metrizable space.
\end{lemma}
\begin{proof} See Theorem 2 (Chapter B) of \cite{hildenbrand2015core}, or Corollary 3.95 of  \cite{aliprantis2006infinite}.\end{proof}

In particular, we shall denote the metric which generates the closed convergence topology by $\delta_C$.  Recall that $X$ is metrizable. Let $d$ be an associated metric.  When $X$ is compact, one can choose $\delta_C$ to be the Hausdorff metric on subsets of $X\times X$ induced by $d$. On the other hand, if $X$ is only locally compact, then $\delta_C$ may be chosen to coincide with the Hausdorff metric on subsets of $X_{\infty}\times X_{\infty}$, where $X_{\infty}$ is the one-point compactification of $X$ together with some metric generating $X_{\infty}$.  See \citet{aliprantis2006infinite} for details.

\subsection{Weak rationalizations}\label{sec:weakrat}
We now present our results on the asymptotic behavior of preference estimates based on finite data. The results generalize those stated in Section~\ref{sec:results}.

For our first result, we must define two notions.
We say that a preference relation $\succeq$ is \df{locally strict} if for every $x,y\in X$ with $x\succeq y$, and every neighborhood $V$ of $(x,y)$ in $X\times X$ there is $(x',y')\in V$ with $x'\succ y'$. 

The first main result gives conditions of convergence of preferences that weakly rationalize the experimental observations. Note that Theorem~\ref{thm:weakrat}
generalizes Theorem~\ref{thm:weakrat-v1}.

\begin{theorem}\label{thm:weakrat}
	Suppose that
	\begin{enumerate}
		\item the subject's preference $\succeq^*$ is continuous and strictly monotone,
		\item the strict partial order $<$ is an open set,
		\item every continuous and strictly monotone preference relation is locally strict.
	\end{enumerate}
	Let $c \sqsubseteq c_{\succeq^*}$ be a choice sequence, and let $\succeq_k$ be a continuous and strictly monotone preference that weakly rationalizes $c^k$. Then, $\succeq_k \rightarrow \succeq^*$ in the closed convergence topology.
\end{theorem}
Note that the assumption that $\succeq^*$ and $\succeq_k$ are transitive is not needed.  Instead, each of these only needs to be continuous, strictly monotone, and complete.

Note that Theorem~\ref{thm:weakrat} requires the existence of the data-generating preference $\succeq^*$.  However, even if existence of this object is not supposed, we can still ``bound'' the set of rationalizations to an arbitrary degree of precision.  This is the point of the next result.

For a choice sequence $c$, let 
$\mathcal{P}^k(c)$ be the set of continuous and strictly monotone preferences that weakly rationalize $c^k$. For a set of binary relations $S$, define $\mbox{diam}(S)=\sup_{(\succeq,\succeq')\in S^2}\delta_C(\succeq,\succeq')$ to be the diameter of $S$ according to the metric $\delta_C$ which generates the topology on preferences.

\begin{theorem}\label{thm:diameter}Suppose that $<$ has open intervals.  Let $c$ be a choice sequence, and suppose that each strictly monotone continuous preference is also locally strict.  Then one of the following holds:
\begin{enumerate}
\item There is $k$ such that $\mathcal{P}^k(c)=\varnothing$.
\item $\lim_{k\rightarrow \infty}\mbox{diam}(\mathcal{P}^k(c))\rightarrow 0$.
\end{enumerate}
\end{theorem}

That is, either a choice sequence is eventually not weakly rationalizable by a strictly monotone preference, \emph{or}, the set of rationalizations becomes arbitrarily small.

Note that, as for Theorem~\ref{thm:weakrat}, Theorem~\ref{thm:diameter} can dispense with the notion of transitivity.  In this case, we would define $\mathcal{P}^k(c)$ to be the set of (potentially nontransitive) complete, continuous, and strongly monotone relations weakly rationalizing $c^k$.

\subsection{Strong rationalizations}\label{sec:strongrat}

Say that the set $X$, together with the collection of finite experiments $\Sigma_\infty$, has the \df{countable order property} if for each $x\in X$ and each neighborhood $V$ of $x$ in $X$ there is $x',x''\in B \cap V$ with $x'\leq x\leq x''$.
We say that $X$ has the \df{squeezing property} if for any convergent sequence $\{x_n\}_n$ in $X$, if $x_n\rightarrow x^*$ then there is an increasing sequence $\{x'_n\}_n$, and an a decreasing sequence $\{x''_n\}_n$, such that $x'_n\leq x_n\leq x''_n$, and  $\lim_{n\rightarrow \infty} x'_n = x^* = \lim_{n\rightarrow \infty} x''_n$.

\begin{theorem}\label{thm:weaklymon}
	Suppose that
	\begin{enumerate}
		\item the subject's preference $\succeq^*$ is weakly monotone,
		\item $(X,\Sigma_\infty)$ has the countable order property, and $X$ the squeezing property.
	\end{enumerate}
	Let $\succeq_k$ be a continuous and weakly monotone preference that strongly rationalizes the choice function of order $k$ generated by $\succeq^*$. Then, $\succeq_k \rightarrow \succeq^*$ in the closed convergence topology.
\end{theorem}

The countable order and squeezing properties are technical but not vacuous. Importantly, as stated below in Proposition~\ref{prop:AAandRn}, they are satisfied for two common cases of interest discussed in Section~\ref{sec:SMv1}. Therefore, Theorem~\ref{thm:weaklymon} generalizes Theorem~\ref{thm:weaklymon-v1}.

\begin{proposition}\label{prop:AAandRn}
	If either
	\begin{enumerate}
		\item the set of alternatives $X$ is $\Re^n$ endowed with the order of weak vector dominance, or
		\item the set of alternatives $X$ is $\Delta([a,b])$ endowed with the order of weak first-order stochastic dominance,
	\end{enumerate}
 	then $X$ has the squeezing property, and there is $\Sigma_\infty$ such that $(X,\Sigma_\infty)$ has the countable order property.
\end{proposition}

One key element behind the above two results is a natural order on the sets of possible alternatives. Via monotonicity, the order adds structure to the families of preferences under consideration. Crucially, the order also relates to the topology on the set $X$.

\subsection{Applications}

We have already highlighted the application of our results to Euclidean consumption spaces and Anscombe-Aumann acts over monetary lotteries (see Section~\ref{sec:results}). Here we discuss two other domains of application of our results.

\subsubsection{Lotteries over a finite prize space.} Let $\Pi$  be  a finite \df{prize space}. The objects of  choice are the elements of $X=\Delta(\Pi)$. Fix a strict ranking of the elements of $\Pi$, and enumerate the elements of $\Pi$ so that $\pi_1$ is ranked above $\pi_2$, which is ranked above $\pi_3$, and so on. Then the elements of $X$ can be ordered with respect to first-order stochastic dominance: $x$ is larger than $y$ in this order if the probability of each set $\{\pi_1,\ldots,\pi_k\}$ is at least as large under $x$ than under $y$, for all $k=1,\ldots,\abs{\Pi}$. A preference over $X$ is monotone if it always prefer larger lotteries over smaller ones.\footnote{The objective order on $\Pi$ is not really needed in this case; see Example~\ref{ex:euexample}. The point of the example is to illustrate~Theorem~\ref{thm:weakrat}.}

Suppose that choices are generated by an expected utility preference $\succeq^*$. The fact that $\succeq^*$ belongs to the expected utility family implies that there are rationalizing expected utility preference $\succeq_k$, for each finite experiment $k$. Then, the above results ensure that these converge to $\succeq^*$. Of course the same would be true of any (monotone and continuous) rationalizing preference: any mode mis-specification would be corrected in the limit. In other words, any arbitrary sequence of rationalization has the data-generating preference $\succeq^*$ as its limit.

\subsubsection{Dated rewards}\label{sec:datedrewards}
We can apply our theory to intertemporal choice. Specifically to the choice of \df{dated rewards} (\cite{fishburndr}).  The set of elements of choice is $\Re^2_+$.  A point $(x,t)\in\Re^2_+$ is interpreted as a  monetary payment of $x$ delivered on date $t$.  Endow $\Re^2_+$ with the order $\leq^i$ whereby $(x,t)\leq^i (x',t')$ if $x\leq x'$ and $t'\leq t$.  Monotonicity of preferences means that more money earlier is  preferred to less money later.

Now one can postulate a preference $\succeq^*$ such that $(x',t')\succeq^*(x,t)$ iff $\delta^t v(x)\leq \delta^{t'} v(x')$, for some $\delta\in (0,1)$ and a strictly increasing function $v:\Re_+\rightarrow \Re$. This means that $\succeq^*$ follows the exponential discounting model. Again, any finite experiment would be rationalizable by exponential preference, and these would converge to the limiting $\succeq^*$.

\section{Identification of Utility Functions}\label{sec:utility}

In this section, we investigate the relation between preferences and utility.  Preferences remain topologized with the closed convergence topology.  We study continuous utility representations, and ask when the identification of a preference allows the identification of a utility (or conversely).  We show that if we endow the set of continuous utility functions with the topology of uniform convergence on compacta, then convergence in one sense is equivalent to convergence in the other.  Formally, we establish that there is a homeomorphism between the two spaces (when we identify two utility functions representing the same preference relation).

Throughout this section, the space of possible alternatives $X$ is connected (and remains a locally compact Polish space, as described in our model).
Connectedness is imposed so that every continuous preference admits a continuous representation, as in \citet{debreu}.

We denote by $\mathcal{U}$ the set of strictly increasing and continuous utility functions on $X$.  Similarly, $\mathcal{R}^{\mbox{mon}}$ denotes the set of preferences which are strictly monotone and continuous.

Suppose the existence of a set $M\subseteq X$, satisfying the following conditions:
\begin{itemize}
	\item $M$ has at least two distinct elements;  $M$ is connected and totally ordered by $<$. In other words $x,y \in M$ and $x\neq y$ implies $x<y$ or $y<x$.
	\item For any $m\in M$ and any neighborhood $U$ of $m$ in $X$ there is
	  $\ul m,\overline m\in M$, with \[m\in [\ul m,\overline m]\subseteq U.\]
	  Moreover if $m$ is not the largest element of $M$ we can choose
	  $\overline m$ such that $m <\overline m$, and if $m$ is not the smallest element
	  we can choose $\ul m$ such that $\ul m< m$.
	\item Any bounded sequence in $X$ is bounded by elements of $M$. That
	  is, for any bounded sequence $\{x_n\}$ there are $\underline{m}$ and
	  $\overline{m}$ and $k$ large so that $\underline{m}\leq x_n \leq
	  \overline{m}$.
\end{itemize}

Let $\Phi:\mathcal{U}\rightarrow\R^{\mbox{mon}}$ such that $\Phi(u)$ is the preference represented by $u\in\mathcal{U}$.\footnote{That is, $x \mathrel{\Phi(u)} y$ if and only if $u(x) \geq u(y)$.}

We provide two examples below that demonstrate the property just mentioned for the case of alternatives of the form $X=\Delta([a,b])$ and $X=\Delta([a,b])^n$.

\begin{example}\label{example:weaktopology2}
	Let $X=\Delta([a,b])$ be the set of Borel probability distributions on a real compact interval $S= [a,b] \subseteq \Re$. Endow $X$ with the weak* topology and let $\leq$ be first-order stochastic dominance.  Observe that $X$ is compact, metrizable, and separable (Theorems 15.11 and 15.12 of \citet{aliprantis2006infinite}).  Observe also that $X$ has the countable order property (see Lemma~\ref{example:weaktopology} in Appendix~\ref{app:proof:prop:AAandRn}).

	Let $<$ be the strict part of $\leq$. Identify $S$ with degenerate probability distributions, so that $s\in S$ denotes the element of $X$ that assigns probability 1 to $\{s\}$, say $\delta_s$.  Let $M = S$. The relative topology on $S$ coincides with the usual topology, so $S$ is connected.  Note that $a\leq x \leq b$ for any $x\in X$.

	Let $m\in M$ and $U$ be a neighborhood of $m$ in $X$. For each $x\in X$, let $F^x$ be the cdf associated to $x$. Choose $\ep$ such that the ball $B_\ep (m)$ (in the Levy metric) with center $m$ and radius $\ep$ is contained in $U$. Let $\ep'<\ep$. Then if $y\in [m-\ep',m+\ep']$ we have that
	\begin{align*}
		F^y(s-\ep)-\ep \leq F^{m-\ep'} (s-\ep) - \ep & < 1 = F^m (s) \text{ if } s-\ep\geq m-\ep' \\
		F^y(s-\ep)-\ep \leq F^{m-\ep'} (s-\ep) - \ep & = -\ep < F^m (s) \text{ if } s-\ep < m-\ep'
	\end{align*}
	Similarly,
	\begin{align*}
		F^m(s) = 0 < F^{m+\ep} (s+\ep) + \ep & \leq F^y(s + \ep) + \ep \text{ if } s + \ep \leq m+\ep' \\
		F^m(s) < 1+\ep = F^{m+\ep'}(s+\ep)+\ep & \leq F^y(s+\ep)+\ep \text{ if } s + \ep > m +\ep'.
	\end{align*}
	These inequalities mean that $y\in B_\ep(m)$. Thus $[m-\ep',m+\ep']\subseteq U$, as $y$ was arbitrary.
\end{example}

\begin{example}
	Let $\Omega$ be a nonempty set such that $|\Omega|<+\infty$.  Suppose $\Omega$ represents a set of \emph{states of the world}.  Then $\Delta([a,b])^{\Omega}$, endowed with the product weak* topology, and ordered by the product order, of $\Omega$ copies of first order stochastic dominance, represents the set of Anscombe-Aumann acts, \citet{anscombe}.  Let $S=\{(\delta_s,\ldots,\delta_s):s\in[a,b]\}$; the constant acts whose outcomes are degenerate lotteries.  Let $M=S$, as in the previous example; and all topological properties satisfied there are also satisfied here.
\end{example}

The following results generalize those derived originally by \citet{mas1974continuous}, who worked with $\Re_+^n$.

\begin{theorem}\label{thm:openmap}
	$\Phi$ is an open map.
\end{theorem}

\begin{theorem}\label{thm:border}(\cite{border1994dynamic} Thm 8)
Let $(X,d)$ be a locally compact and separable metric space and $\R$
be the space of continuous preference relations on $X$, endowed with the
topology of closed convergence. If $\succeq_u = \Phi(u)$ is locally
strict, then $\Phi$ is continuous at $u$. If $M$ has no isolated
points, and $\Phi$ is continuous at $u$, then $\succeq_u$ is locally
strict.
\end{theorem}

Define an equivalence relation $\simeq$ on $\mathcal{U}$ by $u\simeq v$ if there exists $\varphi:\Re\rightarrow\Re$ strictly increasing for which $u=\varphi \circ v$.  Then let $\mathcal{U}/\simeq$ denote the set of equivalence classes of $\mathcal{U}$ under $\simeq$ endowed with the quotient topology; the equivalence class of $u\in\mathcal{U}$ is written $[u]$. The map $\hat{\Phi}:\mathcal{U}/\simeq\rightarrow \mathcal{R}^{\mbox{mon}}$ is defined in the natural way, via $\hat{\Phi}([u])=\Phi(u)$.\footnote{Observe that this function is well-defined.  If $v\in [u]$, then there is strictly increasing $\varphi$ for which $v = \varphi \circ u$, hence $v$ and $u$ represent the same preference.}

\begin{theorem}\label{cor:homeomorphism}
	$\hat{\Phi}$ is a homeomorphism.
\end{theorem}

Given the discussion of Example~\ref{example:weaktopology2}, Theorem~\ref{cor:homeomorphism} generalizes Theorem~\ref{thm:homeomorphism-v1}. The role of $M$ in the case of $\Re^n_+$ is played by the equal-coordinates ray. It is also straightforward to apply Theorem~\ref{cor:homeomorphism} to intertemporal choice by way of the model of dated rewards (see Section~\ref{sec:datedrewards}, by letting $M$ be the line $\{(x,0):x\geq 0\}$. 

\section{Non-monotone preferences and local strictness}\label{sec:expectedutility}

When the set of utility functions is compact, we can obtain a particularly strong result that does not rely on monotonicity, or the existence of a preference relation generating the choices. Instead, the generating preference is obtained endogenously as the limit of rationalizing preferences.\footnote{In particular, these results should be contrasted with the example in Section~\ref{sec:notclosed}.} 

Let $\calv$ be a compact set of continuous functions in the topology of compact convergence, and let $\Phi(\calv)$ denote the image of $\calv$ under $\Phi$, so that $\Phi(u)$ is the preference represented by $u$.

\begin{theorem}\label{thm:expectedutility}Suppose $\calv$ is compact, and that all $\succeq\in\Phi(\calv)$ are locally strict.  Let $c$ be a choice sequence, and let $\succeq_k\in\calv$ weakly rationalize $c^k$. Then, there exists $\succeq^*\in\Phi(\calv)$ such that $\succeq_k \rightarrow \succeq^*$ in the closed convergence topology.  Furthermore, if $\succeq'_k$ also weakly rationalizes $c^k$, then $\succeq'_k\rightarrow\succeq^*$. \end{theorem}

Observe that knowledge of a generating preference $\succeq^*$ is not required; but the hypothesis that there is a weak rationalization $\succeq_k$ for every $c^k$ suggests the possibility.

Theorem~\ref{thm:expectedutility} implies that one can sometimes obtain asymptotically utility rationalizations drawn from $\calv$. In particular, when $\calv$ is compact, $\Phi(\calv)$ consists of locally strict preferences, and  $\Phi$ is a homeomorphism then $\Phi^{-1}(\succeq_k)\in\calv$ converges to a utility for $\succeq^*$ in $\calv$. One application of this kind is in Example~\ref{ex:euexample}.

\begin{example}\label{ex:euexample}Let $X$ be a finite set, and let $\Delta(X)$ be the lotteries on $X$ (topologized as elements of Euclidean space).  Consider the set of nonconstant expected utility preferences.  Then the hypotheses of Theorem~\ref{thm:expectedutility} hold here.  To see this, observe that the set of nonconstant von Neumann-Morgenstern utility indices is homeomorphic to the set \[S=\{u\in \Re^X:\sum_x u_x = 0, \|u\|=1\}.\]  It is straightforward to see that the map $\phi:S\rightarrow C(\Delta(X))$ given by $\phi(u)(p)=\sum_x u_x p(x)$ is continuous.  So, let $\calv=\phi(S)$ which is compact; then the set $\Phi(\calv)$ is the set of nontrivial expected utility preferences.  Finally, observe that each nonconstant expected utility preference is locally strict.  For, if $\succeq$ is nonconstant, then there are $p,q\in\Delta(X)$ for which $p\succ q$.  Then for any $r\succeq s$, for any $\alpha> 0$, $\alpha p + (1-\alpha)r \succ \alpha q +(1-\alpha)s$.  Choose $\alpha$ small to be within any neighborhood of $(r,s)$.
\end{example}

Next, Example~\ref{ex:lipschitzex} allows for an infinite set of prices, but restricts von Neumann-Morgenstern utilities to have lower and upper Lipschitz bounds.
\begin{example}\label{ex:lipschitzex}We can consider $\Re_+^n$, and a class of utility functions $\mathcal{U}_a^b$, where $a,b\in\Re$ with $0<a<b$.  \[\mathcal{U}_a^b=\{u\in C(\Re_+^n):\forall i\wedge \forall (x_i < y_i), a(y_i-x_i)\leq u(y_i,x_{-i})- u(x_i,x_{-i})\leq b(y_i-x_i)\}.\]

Observe that $\mathcal{U}_a^b\subseteq \mathcal{U}$, and consists of those members satisfying a certain Lipschitz property (namely, Lipschitz boundedness above and below).  By the Arzela-Ascoli Theorem (see \citet{Dugundji}, Theorem 6.4), $\mathcal{U}_a^b$ is compact.  Furthermore, each $\succeq\in\Phi(\mathcal{U}_a^b)$ is locally strict, as it is strictly monotonic.\end{example}

\section{Infinite and Countable Data}\label{sec:infinitedata}

In this section, we propose two sufficient conditions that enable the recovery of the subject's preference from its restriction to a countable set of data points.

We first show below that if we can observe a subrelation of a locally strict and continuous binary relation on a dense set, then we can infer the entire binary relation.

\begin{theorem}\label{cor:lsequal} Suppose that $\succeq$ and $\succeq'$ are two complete, continuous, and locally strict binary relations. Let $B\subseteq X$ be dense.  If $\succeq|_{B\times B}\subseteq\succeq'|_{B\times B}$, then $\succeq = \succeq'$.
\end{theorem}

We then make no restriction on the preferences other than continuity, but requires the underlying space of alternatives to be connected.

\begin{theorem}\label{thm:connected}
	Suppose that $\succeq$ and $\succeq'$ are two continuous preference relations. 
 Suppose $X$ is connected, and let $B\subseteq X$ be dense.  If $\succeq|_{B\times B} = \succeq'|_{B\times B}$, then $\succeq = \succeq'$.
\end{theorem}

Note that Theorem~\ref{thm:connected} generalizes Theorem~\ref{thm:connected-v1}. Without connectedness, this result can fail.  A preference $\succeq$ can be increasing on $(0,1)\cup (2,3)$, but there are two possible ways to extend it to $[0,1]\cup [2,3]$; either by setting $1 \sim 2$, or $2\succ 1$.

A classical procedure, attributed to Allais (see \citealp{allais}) allows one to elicit multiple choices with one suitably randomized choice.  Roughly, one uses a randomization device whose outcome is a choice set, and asks a subject to announce what she would choose \emph{ex-ante} from each of the sets in the support in the distribution.  A decision maker who respects basic monotonicity postulates (see \citealp{azrieli}) correctly announces each of their choices.

If we can uncover an entire preference from each of these choices, then we are able to elicit an entire preference using one suitably chosen random device.  Here, we do not investigate this theory in its full generality.  But if there is a countable dense subset of alternatives, and a continuous preference can be inferred from its behavior on a countable dense subset, then we can utilize the Allais mechanism to uncover an entire preference with a single randomized choice.  For example, we would enumerate the pairs of elements from the countable dense subset, say $B_1,B_2,\ldots$, and randomize so that each one realizes with probability $2^{-k}$.

\section{On the meaning of $\succeq^*$}\label{sec:discusspreference}

Some economists are comfortable with the idea that an agent ``has'' a data-generating preference $\succeq^*$, and some are not. The former assume that the preference is something intrinsic to the agent, and that when presented with a choice situation the agent can access his preference and choose accordingly. Under this interpretation, our paper gives conditions under which a finite experiment can approximate, to an arbitrary degree of precision, the underlying preference that the agent uses to make choices.

Other economists argue that preferences are just choices. For those in this position, it is meaningless to speak of a preference over pairs of alternatives from which the agent never chooses. We are sympathetic to this view, and our paper also contributes to this interpretation. Under proper conditions---conditions that we provide in our paper---continuity ``defines'' preferences over $X$ given choices over a countable subset. This is important because estimated preferences provide a guide for making normative recommendations and out of sample predictions. An economist may want to estimate $\succeq^*$ so as to make policy recommendations that are in the agent's interest; in fact, this is a common use of estimated preferences in applied work. Similarly, the economist may want to use $\succeq^*$ as an input in a structural economic model, and thereby make predictions for different configurations of the model. The existence and meaning of $\succeq^*$ is then provided for by the continuity assumption.

Moreover, viewed from this angle, Theorem~\ref{thm:diameter} allows us to say that the set of rationalizations can be made arbitrarily small as more and more data are observed.\footnote{This is true in spite of the claim we make in Section~\ref{sec:notclosed}.  It is true that the set of rationalizations may ``shrink'' to something which is not transitive, but this set is shrinking nonetheless and always contains preference relations (except in the limit).}  In this manner, one can bound errors in welfare statements or out of sample predictions to an arbitrary degree of precision.

We conclude this section with two examples that illustrate the importance of postulating existence of an agent's preference: without the postulate, the inferred preference may otherwise fail to converge.

\subsection{The set of weakly monotone preference relations is not closed}\label{sec:notclosed}

Suppose we are interested in rationality in the form of a strictly monotonic continuous preference relation.  Observe that Theorems~\ref{thm:weakrat} and \ref{thm:weaklymon} hypothesize the existence of $\succeq^*$.  If $\succeq^*\in\Rmon$, for example, then we know that, in the limit, rationalizing relations will be transitive if every $\succeq_k$ is.  Unfortunately, we show in this section, if we do not know that $\succeq^*$ is transitive, we cannot ensure that it is, even if each $\succeq_k$ is.  That is, we demonstrate a sequence $\succeq_k$ of strictly monotone preferences, where $\succeq_k\rightarrow\succeq^*$ in the closed convergence topology, but $\succeq^*$ is not transitive.

The data are rationalizable, but the rationalization requires intransitive indifference.  So the properties of the rationalizations of order $k$ cannot be preserved.

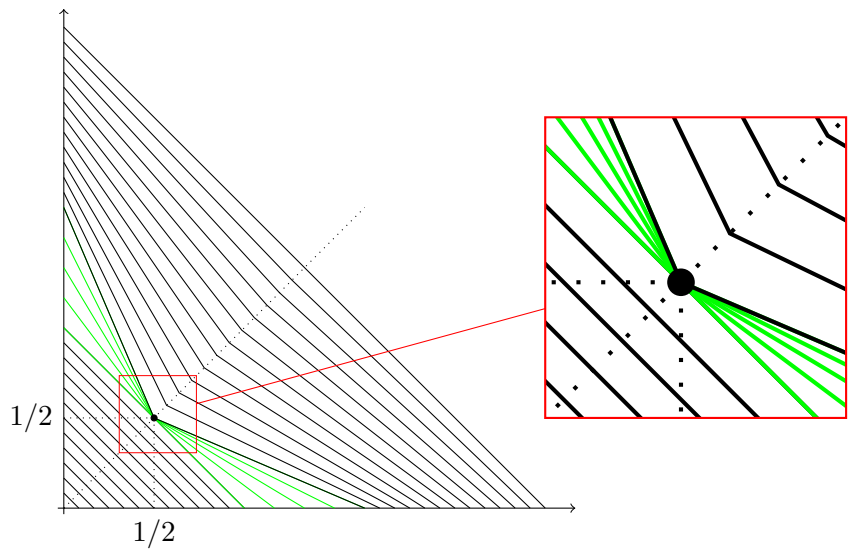
\begin{figure}
\begin{center}
\begin{tikzpicture}[scale=.8, spy using outlines={magnification=3.9, size=4cm, connect spies}]
\draw[->] (0,-.1) -- (0,8.3) node[anchor=west] {}; 
\draw[->] (-.1,0) -- (8.5,0)  node[anchor=north] {};
\draw[dotted] (0,0) -- (5,5);


\foreach \i in {0,.25,...,3}
{
\draw[-,thin] (\i,0) -- (0,\i);
}

\foreach \i in {3,3.5,...,5}
{
\draw[-,thin,green] (\i,0) -- (1.5, 1.5) -- (0,\i);
}

\foreach \i in {5,5.25,...,8}
{
\draw[-,thin] (\i,0) -- (-8/3 + 2.5*\i/3,-8/3 + 2.5*\i/3) -- (0, \i);
}

\draw [fill] (1.5,1.5) circle (.05cm);

\draw [thin, dotted] (0,1.5) node[anchor=east] {\small $1/2$}-- (1.5,1.5);
\draw [thin, dotted] (1.5,0) node[anchor=north] {\small $1/2$} -- (1.5,1.5);

  \spy [red] on (1.25,1.25)
             in node [right] at (8,4);

\end{tikzpicture}\end{center}
\caption{A non-transitive preference}\label{fig:nontrans}
\end{figure}

Figure~\ref{fig:nontrans} exhibits a non-transitive relation. The
example is taken from \cite{grodal1974note}. The
lines depict indifference curves, but all the green indifference
curves intersect at one point: $(1/2,1/2)$. This makes the preference
non-transitive; specifically the indifference part of the preference
would be intransitive here.

Now imagine a collection of binary comparisons that do not include
$(1/2,1/2)$. Suppose that this collection is the limit of a finite
number of binary comparisons, making it at most countable. There must
exist a ball around $(1/2,1/2)$ that does not include any of the
comparisons. Consider the diagram in Figure~\ref{fig:nontransball}.
The preferences have been modified close to $(1/2,1/2)$ so that
transitivity holds.

\begin{figure}
\begin{center}
\begin{tikzpicture}[scale=.8, spy using outlines={magnification=3.9, size=4cm, connect spies}]
\draw[->] (0,-.1) -- (0,8.3) node[anchor=west] {}; 
\draw[->] (-.1,0) -- (8.5,0)  node[anchor=north] {};
\draw[dotted] (0,0) -- (5,5);


\draw[-,green,name path=m1] (5,0) -- (1.5, 1.5) ;
\draw[-,green,name path=m2] (1.5, 1.5) -- (0,5);
\draw[-,green,name path=n1] (3.8,0) -- (1.5, 1.5) ;
\draw[-,green,name path=n2] (1.5, 1.5) -- (0,3.8);
\draw [name path=arcone] (1.8,1.5) arc (0:180:.3cm);
\draw [name path=arctwo] (1.8,1.5) arc (360:180:.3cm);
\draw [fill, white] (1.5,1.5) circle (.3cm);

\foreach \i in {0,.25,...,3}
{
\draw[-,thin] (\i,0) -- (0,\i);
}

\foreach \i in {5.25,5.5,...,8}
{
\draw[-,thin] (\i,0) -- (-8/3 + 2.5*\i/3,-8/3 + 2.5*\i/3) -- (0, \i);
}

\draw [blue,thin,fill, fill opacity=.2]  (1.5,1.5) circle (.3cm);

\path [name intersections={of=n1 and arctwo,by=l1}];
\path [name intersections={of=n2 and arcone,by=l2}];

\path [name intersections={of=m1 and arctwo,by=o1}];
\path [name intersections={of=m2 and arcone,by=o2}];

\draw [green] (l1) -- (l2);\draw [green] (o1) -- (o2);

\draw [fill] (1.5,1.5) circle (.05cm);

  \spy [red] on (1.25,1.25)
             in node [right] at (8,4);
\end{tikzpicture}\end{center}
\caption{A transitive preference}\label{fig:nontransball}
\end{figure}
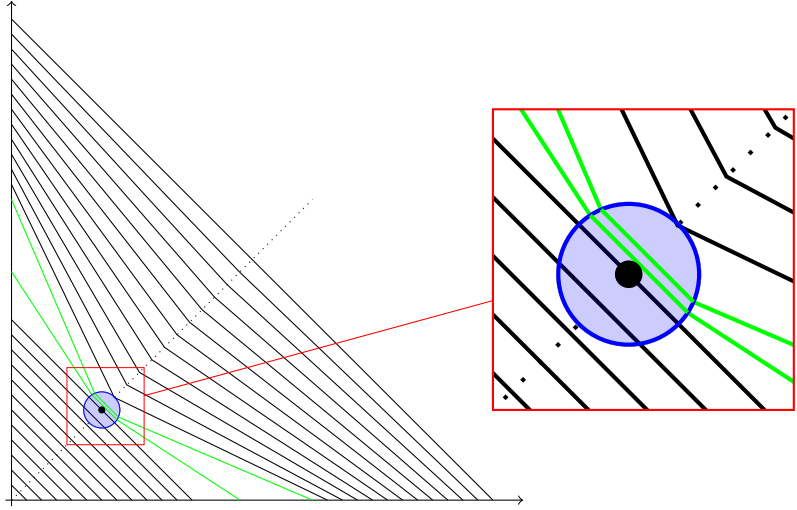

This example is not particularly troubling, however.  First, with finite experimentation, the violation of transitivity will never be ``reached.''  Second, the violation here is not particularly egregious.  Only transitivity of indifference is violated.  This holds quite generally. It can be shown that any limit point of a sequence of preference relations must be quasitransitive, so that whenever $x\succ y$ and $y \succ z$, it follows that $x\succ z$.\footnote{The argument is in \citet{grodal1974note}, but to see this suppose that $\succeq^n\rightarrow \succeq$, where each $\succeq^n$ is a preference relation.  It can be shown that $\succeq$ is complete, so suppose by means of contradiction that there are $x,y,z\in X$ for which $x\succ y$, $y \succ z$, and $z \succeq x$.  So, there are $x^n$, $z^n$ for which $z^n \succeq^n x^n$, $x^n\rightarrow x$, and $z^n\rightarrow z$.  For each $n$, either $z^n \succeq^n y$ or $y \succeq^n x$, so that without loss, there is a sequence for which $z^n \succeq^n y$, \emph{i.e.} $z \succeq y$, a contradiction.}  Quasitransitive relations enjoy many of the useful properties of preferences.  For example, continuous quasitransitive relations possess maxima on compact sets, see \emph{e.g.} \citet{bergstrom1975}.

\subsection{The set of locally strict relations is not closed}
Finally we present an example to show that the set of locally strict preference relations is not closed.
Let $X = [-3,-1] \cup [1,3]$.  For each $n$, let $u_n(x)=-(x+2)^2 + \frac{1}{n}$ on $[-3,-1]$ and $u_n(x)=(x-2)^2 -\frac{1}{n}$ on $[1,3]$. See Figure~\ref{fig:LSnotclosed}. The function $u_n$ represents a locally strict relation $\succeq_n$.

Let $u^*(x)$ be the pointwise limit of $u_n$; i.e. $u^*(x)=-(x+2)^2$ on $[-3,-1]$ and $u^*(x) = (x-2)^2$ on $[1,3]$. The function $u^*$ represents $\succeq^*$ which is \emph{not} locally strict.  Observe that $-2 \succeq^* 2$, but for small neighborhoods there is no strict preference.

However, it is also straightforward by checking cases to show that $\succeq_n\rightarrow \succeq^*$.

\begin{figure}\begin{center}
   \begin{tikzpicture}
    \def\efleft{\x,{ -(\x+2)*(\x+2)}}
    \def\efright{\x,{(\x-2)*(\x-2) }}
    \def\efmiddle{\x,{ \x}}
  \draw[color=blue,domain=-3:-1] plot (\efleft) node[right] {};
  \draw[color=blue,domain=1:3] plot (\efright) node[right] {};
  \draw[color=blue,domain=-1:1] plot (\efmiddle) node[right] {};

    \def\efleft{\x,{ -(\x+2)*(\x+2)+1/9}}
    \def\efright{\x,{(\x-2)*(\x-2) - 1/9}}
    \def\efmiddle{\x,{ (1-1/9)*\x}}
  \draw[color=green,domain=-3:-1] plot (\efleft) node[right] {};
  \draw[color=green,domain=1:3] plot (\efright) node[right] {};
  \draw[color=green,domain=-1:1] plot (\efmiddle) node[right] {};

    \def\efleft{\x,{ -(\x+2)*(\x+2)+1/3}}
    \def\efright{\x,{(\x-2)*(\x-2) - 1/3}}
    \def\efmiddle{\x,{ (1-1/3)*\x}}
  \draw[color=red,domain=-3:-1] plot (\efleft) node[right] {};
  \draw[color=red,domain=1:3] plot (\efright) node[right] {};
  \draw[color=red,domain=-1:1] plot (\efmiddle) node[right] {};

    \draw[->] (-3.04,0) -- (3.2,0) node[right] {};
    \draw[->] (0,-2.04) -- (0,3) node[above] {};
\end{tikzpicture}\end{center}   \caption{The set of locally strict preferences is not closed.}
  \label{fig:LSnotclosed}
\end{figure}
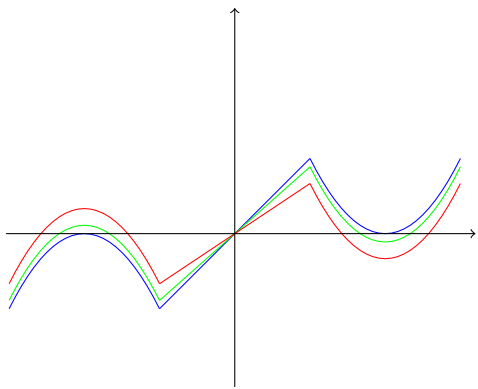

\newpage

\bibliographystyle{ecta}
\bibliography{identification}

\newpage

\appendix


\section{About Closed Convergence}\label{sec:closedconvergence}

We recall below the formal definition of closed convergence, used throughout the results of this paper. Let $\mathcal{F}=\{F^n\}_n$ be a sequence of closed sets in $X \times X$. We define $\Li(\mathcal{F})$ and $\Ls(\mathcal{F})$ to be closed subsets of $X \times X$ as follows:
\begin{itemize}
	\item $(x,y) \in \Li(\mathcal{F})$ if and only if, for all neighborhood $V$ of $(x,y)$, there exists $N \in \Na$ such that $F^n \cap V \neq \os$ for all $n \geq N$.
	\item $(x,y) \in \Ls(\mathcal{F})$ if and only if, for all neighborhood $V$ of $(x,y)$, and all $N \in \Na$, there is  $n \geq N$ such that $F^n \cap V \neq \os$.
\end{itemize}
Observe that $\Li(\mathcal{F})\subseteq \Ls(\mathcal{F})$. The definition of closed convergence is as follows.
\begin{definition}
	$F^n$ converges to $F$ in the \df{topology of closed convergence} if $\Li(\mathcal{F})=F=\Ls(\mathcal{F})$.
\end{definition}

\section{Proof of Proposition~\ref{prop:indiff}}\label{app:proof:proofofpropindiff}

Denote by $(a',b')$ the open interval $\{z\in\Re^n:a'\ll z\ll b'\}$.  For each $k$, let $u_k:\cup_{l=1}^k B_l\rightarrow [0,1]$ be a utility representation of $\succeq^*$ on $\cup_{l=1}^k B_l$.

For each $k$, let $\{[a_i,  b_i]\}_{i=1}^{n_k}$ be a sequence of intervals in $\Re^n$ with the properties that a) $[a,b]\subseteq \cup _{i=1}^{n_k} [a_i,  b_i]$, b) $(a_i,  b_i)\cap (a_j,b_j) = \os$ for $i\neq j$, c) each element of $\cup_{l=1}^k B_l$ is contained in a set $(a_i,  b_i)$, and no two elements of $\cup_{l=1}^k B_l$ are contained in the same, and d) $[a_i,b_i]$ is contained in some ball of radius $(2k)^{-1}$.\footnote{It is obvious that such a sequence exists. First, it is immediate that it exists for $n=1$. For $n>1$ project each $B_k$ onto each of its coordinate and carry out the one-dimensional construction (choosing a sufficiently small radius for the  balls covering each interval). Then take the cartesian product of each one-dimensional interval.}

For each interval  $[a_i,  b_i]$ there is a continuous function  $f_i$ such that $f(x)=0$ for all $x\in [a_i,  b_i]\setminus (a_i,  b_i)$, $f(x) = u_k(x)$ if $x\in (a_i,b_i)\cap \cup_{l=1}^k B_l$, $\sup \{f(x):x\in [a_i,b_i]\} = 2$ and  $\inf \{f(x):x\in [a_i,b_i]\} = -2$. Let $u^*_k:[a,b]\rightarrow \Re$ be the function that coincides with $f_i$ on each $[a_i,b_i]$. Let $\succeq_k$ be the preference relation represented by $u^*_k$, and note that $\succeq_k$ strongly rationalizes the choice function of order $k$ generated by $\succeq^*$,  and is continuous.

Let $x,y\in X$. For each $k$, suppose that $x\in [a_i,b_i]$ for the $k$th sequence of subintervals.
Let $x_k\in [a_i,b_i]$ be such that $u^*_k(x_k)=2$. Note that $\norm{x-x_k}<1/k$. Similarly, suppose that $y\in [a_j,b_j]$ for the $k$th sequence of subintervals and let $y_k\in [a_j,b_j]$ be such that $u^*_k(y_k)=-2$. Then $x_k\succ_k y_k$. Since $(x_k,y_k)\rightarrow (x,y)$ and $x,y\in X$ were arbitrary this means that $\succeq_k\rightarrow X\times X$.

\section{Proof of Proposition~\ref{prop:AAandRn}}\label{app:proof:prop:AAandRn}

The proof is implied by the following lemmas.

\begin{lemma}\label{lem:supsinfsRn}
	Let $X\subseteq\Re^n$. If $\{x'_n\}$ is an increasing sequence in $X$, and $\{x''_n\}$ is a decreasing sequence, such that $\sup\{x'_n:n\geq 1\} = x^* = \inf\{x''_n:n\geq 1\}$. Then \[\lim_{n\rightarrow \infty} x'_n = x^* = \lim_{n\rightarrow \infty} x''_n.\]
\end{lemma}
\begin{proof}
	This is obvously true for $n=1$. For $n>1$, convergence and sups and infs are obtained component-by-component, so the result follows.
\end{proof}

\begin{lemma} \label{lem:boundsequRn}
	Let $X\subseteq\Re^n$. Let $\{x_n\}$ be a convergent sequence in $X$, with $x_n\rightarrow x^*$. Then there is an increasing sequence $\{x'_n\}$ and an a decreasing sequence $\{x''_n\}$ such that $x'_n\leq x_n\leq x''_n$, and $\lim_{n\rightarrow \infty} x'_n = x^* = \lim_{n\rightarrow \infty} x''_n$.
\end{lemma}
\begin{proof}
	Suppose that $x_n\rightarrow x^*$. Define $x'_n$ and $x''_n$ by
	\[x'_n  = \inf \{x_m : n\leq m \} \text{ and } x''_n = \sup \{ x_m : n\leq m \}\]
	Then it is clear that  $x'_n\leq x_n\leq x''_n$, that $x'_n$ is
	increasing, and that $x''_n$ is decreasing. Moreover,
	\begin{align*}
		\lim_{n\rightarrow \infty} x'_n & =  \sup \{ \inf \{x_m : n\leq m \} : n\geq 1\} \\
		&  = x^* \\ & = \inf \{ \sup \{x_m
		: n\leq m \} : n\geq 1\} = \lim_{n\rightarrow \infty} x''_n
	\end{align*}  by
	Lemma~\ref{lem:supsinfsRn}.
\end{proof}

\begin{lemma}
	Let $X = \Delta([a,b])$. Let $\{x_n\}$ be a convergent sequence in $X$, with $x_n\rightarrow x^*$. Then there is an increasing sequence $\{x'_n\}$ and an a decreasing sequence $\{x''_n\}$ such that $x'_n\leq x_n\leq x''_n$, and $\lim_{n\rightarrow \infty} x'_n = x^* = \lim_{n\rightarrow \infty} x''_n$.
\end{lemma}
\begin{proof}
	The set $X$ ordered by first order stochastic dominance is a complete lattice (see, for example, Lemma 3.1 in \cite{kertz2000}). Suppose that $x_n\rightarrow x^*$. Define $x'_n$ and $x''_n$ by
	$x'_n = \inf\{x_m : n\leq m \}$ and $x''_n = \sup\{x_m : n\leq m \}$. Clearly, $\{x'_n\}$ is an increasing sequence, $\{x''_n\}$ is decreasing, and $x'_n\leq x_n\leq x''_n$.

	Let $F_x$ denote the cdf associated with $x$. Note that  $F_{x''_n} (r)  = \inf \{F_{x_m}(r) : n\leq m \}$ while $F_{x'_n} (r)$ is the right-continuous modification of $\sup \{F_{x_m}(r) : n\leq m \}$.
	For any point of continuity $r$ of $F$, $F_{x_m}(r) \rightarrow F_{x^*}(r)$, so \[
	F_x(r) = \sup \{ \inf \{F_{x_m}(r) : n\leq m \} : n\geq 1\} \] by Lemma~\ref{lem:supsinfsRn}.

	Moreover, $F_{x^*}(r) = \inf \{ \sup \{F_{x_m}(r) : n\leq m \} : n\geq 1\}$. Let $\ep>0$. Then
	\[F_{x^*} (r-\ep) \leftarrow  \sup \{F_{x_m}(r-\ep) : n\leq m \} \leq
	F_{x'_n} (r) \leq  \sup \{F_{x_m}(r+\ep) : n\leq m \} \rightarrow F_{x^*}(r+\ep)
	\] Then $F_{x'_n} (r) \rightarrow F_{x^*}(r)$, as $r$ is a point of continuity of $F_{x^*}$.
\end{proof}

The following lemma is immediate.
\begin{lemma}Let $X = \Re_+^n$ with the standard vector order $\leq$, and let $B=\Qe_+^n$.  Then the countable order property is satisfied.\end{lemma}

Our last lemma is a direct implication of Theorem 15.11 of \citet{aliprantis2006infinite}.
\begin{lemma}\label{example:weaktopology}Let $a,b\in\Re$, where $a<b$. Let $X =\Delta([a,b])$, the set of Borel probability distributions on $[a,b]$ endowed with the weak* topology.  Let $B$ be the set of probability distributions $p$ with finite support on $\Qe\cap [a,b]$, where for all $q\in\Qe\cap [a,b]$, $p(q)\in\Qe$.  Then the countable order property is satisfied.\end{lemma}

\section{Proof of Theorems~\ref{thm:weakrat}, \ref{cor:lsequal}, \ref{thm:connected} and \ref{thm:diameter}}

In this section, we let $\Rbmon$ denote the set of complete, continuous, and strictly monotonic binary relations.  Members of $\Rbmon$ need not be transitive. Likewise, $\Rbls$ is the set of complete, continuous, and locally strict binary relations.

We record the following fact:

\begin{lemma}\label{lem:cont} Let $\succeq$ be a continuous binary
  relation. If $x\succ y$ then there are neighborhoods $V_x$ of $x$
  and $V_y$ of $y$ such that $x'\succ y'$ for all $x'\in V_x$ and
  $y'\in V_y$.
\end{lemma}

We now prove Theorems \ref{cor:lsequal} and \ref{thm:connected}.

\begin{proof}[Proof of Theorem~\ref{cor:lsequal}]Follows directly from Lemma~\ref{lem:partialobservability}, below. \end{proof}

\begin{proof}[Proof of Theorem~\ref{thm:connected}]
First, it is straightforward to show that $x \succ y$ implies $x \succeq' y$.  Because otherwise there are $x,y$ for which $x\succ y$ and $y\succ' x$.  Take an open neighborhood $U$ about $(x,y)$ and a pair $(z,w)\in U\cap (B\times B)$ for which $z\succ w$ and $w\succ' z$, a contradiction.  Symmetrically, we also have $x \succ' y$ implies $x\succeq y$.

Now, without loss, suppose that there is a pair $x,y$ for which $x\succ y$ and $x \sim' y$.  By connectedness and continuity, $V=\{z: x \succ z \succ y\}$ is nonempty and by continuity it is open.\footnote{The argument for nonemptiness is as follows.  If, by means of contradiction, $V=\varnothing$, then $\{z: x\succ z\}$ and $\{z:z \succ y\}$ are nonempty open sets.  Further, for any $z\in X$, either $x\succ z$ or $z \succ y$ (because if $\neg (x\succ z)$ then by completeness $z\succeq x$, which implies that $z\succ y$).  Conclude that $\{z:x \succ z\}\cup \{z :z \succ y\}=X$ and each of the sets are nonempty and open (by continuity); these sets are disjoint,
violating connectedness of $X$.}  We claim that there is a pair $(w,z)\in (V\times V)\cap (B\times B)$ for which $w \succ z$.  By denseness of $B$, there is $w\in V\cap B$ for which $x \succ w \succ y$.  Similarly, $\{z:w\succ z \succ y\}$ is nonempty and open; so there is $z\in B$ for which $x \succ w \succ z \succ y$.


We have shown that there is $(w,z)\in (V\times V)\cap (B\times B)$ for which $w\succ z$, so that $x \succ w \succ z \succ y$.  Further, we have hypothesized that $x \sim' y$.  By the first paragraph, we know that $x \succeq' w \succeq' z \succeq' y$.  If, by means of contradiction, we have $w\succ' z$, then $x \succ' y$, a contradiction.  So $w \sim' z$ and $w\succ z$, a contradiction to $\succeq_{B\times B}=\succeq'_{B\times B}$. \end{proof}


\begin{lemma}\label{lem:inclusion}Let $A \subseteq X\times X$.  Then
  $\{\succeq: A \subseteq \succeq \}$ is closed in the closed
  convergence topology.\end{lemma}

\begin{proof}  Let $\succeq_n$ be a convergent sequence in the set in
  question, where $\succeq_n \rightarrow \succeq$.  Then for all
  $(x,y) \in A$, we have $x \succeq_n y$, hence $x \succeq y$.  So
  $(x,y)\in \succeq$.\end{proof}

\begin{lemma}\label{lem:monotoneiscompact} Suppose $X$ is locally
  compact Polish, and that $<$ has open intervals.   Then $\Rbmon$ is closed in the topology of closed convergence.
\end{lemma}

\begin{proof}By Lemma~\ref{lem:closedconvergencecompact}, since $X$ is locally compact
  Polish, the topology of closed convergence is compact metrizable.

Suppose $\succeq_n\rightarrow
  \succeq$ where each $\succeq_n$ is continuous, strictly monotonic, and
  complete.  We know that $\succeq$ is continuous by compactness.  Suppose by means of contradiction that $\succeq$ is not
  strictly monotonic, so that there are $x,y\in X$ for which $x > y$ and $y
  \succeq x$.  Then there are $(x_n,y_n)\rightarrow (x,y)$ for which
  $y_n \succeq_n x_n$.  For $n$ large, $x_n > y_n$, a contradiction to
  the fact that $\succeq_n$ is strictly monotonic.  Finally,
  completeness follows as for each $x,y$, either $x \succeq_n y$ or $y
  \succeq_n x$, so there is a subsequence $n_k$ for which either $x
  \succeq_{n_k} y$ or for which $y \succeq_{n_k} x$.
\end{proof}

\begin{lemma}\label{lem:partialobservability}Suppose that $B$ is dense, $\succeq'$ is complete, and each of $\succeq$ and $\succeq^*$ are continuous and locally strict complete relations.  Then if \[\succeq'|_{B\times B}\subseteq \succeq^*|_{B\times B} \cap \succeq|_{B\times B},\] it follows that $\succeq^*=\succeq$.\end{lemma}

\begin{proof}Suppose, by means of contradiction and without loss of generality, that there are $x,y\in X$ for which $x \succeq^* y$ and $y\succ x$.  By continuity of $\succeq$ and local strictness of $\succeq^*$, we can without loss of generality assume that $x \succ^* y$ and $y \succ x$.  By continuity of each of $\succeq$ and $\succeq^*$, there exists $a,b\in B$ such that $a \succ^* b$ and $b \succ a$.  But by completeness of $\succeq'$, either $a \succeq' b$, contradicting $\succeq'|_{B\times B}\subseteq \succeq|_{B\times B}$, or $b \succeq' a$, contradicting $\succeq'|_{B\times B} \subseteq \succeq^*|_{B\times B}$.\end{proof}

We now turn to the main proof of the theorem.

\begin{proof}[Proof of Theorem~\ref{thm:weakrat}]
By Lemma~\ref{lem:monotoneiscompact}, $\Rbmon$ is
compact.  Let $\succeq'$ be any strictly monotonic and complete binary relation such that for all $k$ and all $\{x,y\}\in\Sigma_k$, $x\in c^k(\{x,y\})$ if and only if $x \succeq' y$ ($\succeq'$ exists by the projection requirement on choice sequences, and by the fact that $c\sqsubseteq c_{\succeq^*}$).

For each $k$, let $\succeq'_k = \{(x,y): \{x,y\}\in \{B_1,\ldots,B_k\}$ and $x \succeq' y\}$.

For each $k$, let \[P_k = \{\succeq\in\Rbmon:
\succeq'_k \subseteq \succeq\},\] the set of relations which weakly rationalize $c$.
Observe by definition that by Lemma~\ref{lem:inclusion}, $P_k$ is closed, and hence
compact.  By assumption, each $\succeq\in P_k$ satisfies $\succeq\in\Rbls$, and obviously, for all $k$, $\succeq^*\in P_k$.
Further, observe that $\bigcap_k P_k = \{\succeq^*\}$, since if $\succeq\in\bigcap_k P_k$, by definition
$\succeq'_{B\times B}\subseteq \succeq^*|_{B\times B}\cap\succeq|_{B\times B}$
and Lemma~\ref{lem:partialobservability}.

The result now follows as each $P_i$ is compact and $\bigcap_k P_k
=\{\succeq^*\}$.  That is, let $\succeq_k \in P_k$, which is a
decreasing, nested collection of compact sets.  Suppose by means of
contradiction and without loss that $\succeq_k \rightarrow
\succeq'\neq\succeq^*$, and observe then that it follows that
$\succeq' \in P_k$ for all $k$, contradicting $\bigcap_i P_i =
\{\succeq^*\}$.
\end{proof}

\begin{proof}[Proof of Theorem~\ref{thm:diameter}]

Observe that for any $k$, the set \[P_k=\{\succeq\in\Rbmon:\succeq \mbox{ weakly rationalizes }c^k\}\] is closed, and hence compact by Lemma~\ref{lem:inclusion}.  Observe that $\mathcal{P}^k(c)\subseteq P_k$.  Moreover, it is obvious that $P_{k+1}\subseteq P_k$.  Suppose that there is no $k$ for which $\mathcal{P}^k(c)= \varnothing$.  Then, since each $P_k\neq\varnothing$ and each $P_k$ is compact, $\bigcap_k P_k \neq\varnothing$.  Let $\succeq^*\in\bigcap_k P_k$.

We claim that $\bigcap_k P_k = \{\succeq^*\}$.  Suppose by means of contradiction that there is $\succeq\neq \succeq^*$ where $\succeq\in\bigcap_k P_k$.  Let $\succeq'$ be any complete relation such that for all $(a,b)\in B\times B$, $a \succeq' b$ if and only if $a\in c^k(\{a,b\})$, for $k$ such that $\{a,b\}\in\Sigma^k$.  Then, by definition of weak rationalization, we have $\succeq'_{B\times B}\subseteq \succeq_{B\times B}\cap \succeq^*_{B\times B}$.  Appeal to Lemma~\ref{lem:partialobservability} to conclude that $\succeq=\succeq^*$, a contradiction.

Finally, since $\bigcap_k P_k = \{\succeq^*\}$, and each $P_k$ is compact, it follows that $\lim_{k\rightarrow \infty}\mbox{diam}(P_k)\rightarrow 0$.\footnote{Otherwise, we could choose $\epsilon > 0$ and two subsequences $\succeq_{k_l}, \succeq'_{k_l}$ such that $\delta_C(\succeq_{k_l},\succeq'_{k_l})\geq \epsilon$ and $\succeq_{k_l}\rightarrow \succeq\in\bigcap_k P_k$ and $\succeq'_{k_l}\rightarrow \succeq'\in \bigcap_k P_k$ where $\delta_C(\succeq,\succeq')\geq \epsilon$, a contradiction.}  Hence, since $0 \leq \mbox{diam}(\mathcal{P}^k(c))\leq \mbox{diam}(P_k)$, the result follows.

\end{proof}



\section{Proof of Theorem~\ref{thm:weaklymon}}

The set of weakly monotone and continuous binary relations is compact in the topology of closed convergence. Suppose wlog that $\succeq^k\rightarrow \succeq$. Then
$\succeq$ is a continuous binary relation. We shall prove that $\succeq=\succeq^*$.

First we show that $x\succ^* y$ implies that $x\succ y$. So let $x\succ^* y$.  Let $U$ and $V$ be neighborhoods of $x$ and $y$, respectively, such that $x'\succ^* y'$
for all $x'\in U$ and $y'\in V$. Such neighborhoods exist by the continuity of $\succeq^*$. We prove first that if $(x',y')\in U\times V$, then there exists $N$
such that $x'\succ_n y'$ for all $n\geq N$. By hypothesis, there exist $x''\in U\cap B$ and $y''\in V\cap B$ such that $x''\leq x'$ and $y'\leq y''$. Each $\succeq_n$ is a strong rationalization of the finite experiment of order $n$, so if $\{\tilde x,\tilde y\}\in\Sigma_n$ then $\tilde x\succ_n \tilde y$ implies that $\tilde x\succ_m \tilde y$ for all $m\geq n$. Since $x'',y''\in B$, there is  $N$ is such that $\{x'',y''\}\in\Sigma_N$. Thus $x''\succ^* y''$ implies that  $x''\succ_n y''$ for all $n\geq N$. So, for $n\geq N$,  $x'\succ_n y'$, as $\succeq_n$ is weakly monotone.

Now we establish that $x\succ y$. Let $\{(x_n,y_n)\}$ be an arbitrary sequence with $(x_n,y_n)\rightarrow (x,y)$. By hypothesis, there is an increasing sequence $\{x'_n\}$, and a decreasing sequence  $\{y'_n\}$, such that $x'_n\leq x_n$ and $y_n\leq y'_n$ while $(x,y)= \lim_{n\rightarrow\infty }(x'_n,y'_n)$.

Let $N$ be large enough that $x'_N\in U$ and $y'_N\in V$. Let $N'\geq N$ be such that $x'_N\succ_{n} y'_N$ for all $n\geq N'$ (we established the existence of such $N'$ above). Then, for any $n\geq N'$ we have that \[x_n\geq x'_n \geq x'_N \succ_n y'_N \geq y'_n\geq y_n.\] By the weak monotonicity of $\succeq_n$, then, $x_n\succ_n y_n$. The sequence $\{(x_n,y_n)\}$ was arbitrary, so $(y,x)\notin \succeq= \lim_{n\rightarrow\infty}\succeq_n$. Thus $\neg (y\succeq x)$. Completeness of $\succeq$ implies that  $x\succ y$.

In second place we show that if $x\succeq^* y$ then $x\succeq y$, thus completing the proof.  So let $x\succeq^* y$.  We recursively construct  sequences $x^{n_k},y^{n_k}$ such that $x^{n_k}\succeq^{n_k}y^{n_k}$ and $x^{n_k}\rightarrow x$, $y^{n_k}\rightarrow y$.

So, for any $k\geq 1$, choose $x'\in N_x(1/k)\cap B$ with $x'\geq x$, and $y'\in N_y(1/k)\cap B$ with $y'\leq y$; so that $x'\succeq^* x\succeq^* y\succeq^* y'$, as $\succeq^*$ is weakly monotone. Recall that
$\succeq_n$ strongly rationalizes $c_{\succeq^*}$ for $\Sigma_n$. So $x'\succeq^* y'$ and $x',y'\in B$ imply that  $x'\succeq_n y'$ for all $n$ large enough. Let  $n_k > n_{k-1}$ (where we can take $n_0=0$) such that  $x'\succeq_{n_k} y'$; and let $x^{n_k}=x'$ and $y^{n_k}=y'$.

Then we have $(x^{n_k},y^{n_k})\rightarrow (x,y)$ and $x_{n_k} \succeq_{n_k} y_{n_k} $. Thus $x\succeq y$.

\section{Proof of Theorem~\ref{thm:openmap} and Proposition~\ref{prop:lowersemicontinuity}}

We begin with two useful lemmas.

\begin{lemma}\label{lem:openmap1}
$\Phi$ is an open map if for any  $u^*\in \mathcal{U}$
  and any sequence $\succeq_n$ in $\R$ with  $\succeq_n \rightarrow
  \Phi(u^*)$, there is a sequence $\{u_n\}$ in $\U$ such that
  $u_n\in\Phi^{-1}(\succeq_n)$ and $u_n \rightarrow u^*$ in the
  topology of compact convergence.
\end{lemma}

\begin{proof}
  Suppose that there is $V\subseteq \mathcal{U}$ open, but $\Phi(V)$ is
  not open.  Then there is $u^*\in V$ and $\succeq_n\notin \Phi(V)$ such
  that $\succeq_n \rightarrow \Phi(u^*)$ (since closed convergence
  topology is metrizable).  Since $u^* \in V$, any sequence $u_n \in
  \Phi^{-1}(\succeq_n)$ for which $u_n \rightarrow u^*$ eventually has
  $u_n \in V$.  But if $u_n$ is chosen to represent $\succeq_n$, this
  implies that $\Phi(u_n) \in \Phi(V)$ for $n$ large, a contradiction.
\end{proof}

\begin{lemma}\label{lem:openmap2}
For any $\succeq$ and $x\in X$, there is a unique $m^*(x)\in M$ with
$x\sim m^*(x)$.    Moreover, if we fix $u\in\U$ then the function
$u_\succeq: X\rightarrow \Re$ defined by   $u_\succeq (x) =
u(m^*(x))$ is a continuous utility representation of $\succeq$.
\end{lemma}
\begin{proof}
   Consider the sets  $A=\{m\in M: m\succeq x\}$ and $B=\{m\in M: x\succeq
   m\}$. These sets are closed because $\succeq$ is continuous,
   their union is $M$ as $\succeq$ is complete, and they are nonempty
   as $\succeq$ is monotone and there exist $\ul m,\overline m\in M$ with $\ul m\leq x \leq \overline m$
   by our hypothesis on $M$.
   $M$ is connected, so $A$ and $B$  cannot be disjoint; hence there is $m\in
   M$ with $x\sim m$. This $m$ must be unique because $M$ is totally
   ordered, and $\succeq$ is strictly monotone.

  We now show that $u_\succeq$ is a continuous utility representation of
  $\succeq$. Let $x\succeq y$. Then transitivity and monotonicity of
  $\succeq$ imply that $m^*(x)\geq m^*(y)$. Thus $u_\succeq(x)=u^*(m^*(x))\geq
  u^*(m^*(y))=u_\succeq(y)$. The converse implications hold as well; thus $u_\succeq$
  represents $\succeq$.

  To prove  continuity, let $x_n\rightarrow x^*$. We shall prove that
  $m_n=m^*(x_n)\rightarrow m^*(x^*) = \hat m$. Suppose first that $\hat m$ is
  not the largest or the least element of $M$. For each neighborhood $U$ of
  $\hat m$ there exists, by our hypothesis on $M$, $\ul m,\overline m\in M$
  with $\ul m<\hat m<\overline m$ and $[\ul m,\overline m]\subseteq U$. Then \[ V=
  \{z\in X : \overline m\succ z \} \cap \{z\in X : z \succ \ul m \}
  \] is a neighborhood of $x^*$, as $x^*\sim \hat m$ and $\succeq$ is
  continuous and monotone. For large enough $n$ then $x^n\in V$, so
  $m^n\in [\ul m,\overline m]\subseteq U$.
  Suppose now that $\hat m$ is the largest element of $M$. Then,
  reasoning as above,
  $x^n\in \{z\in X : z \succ \ul m \}$ for all large enough $n$, so that
  $\ul m\leq m^n$. We have $m^n\leq \overline m$ as $\overline m$ is the largest
  element of $M$. Thus $m^n\in [\ul m,\overline m]\subseteq U$. The argument
  when $\ul m$ is the least element of $M$ is analogous.
\end{proof}

We now turn to the main proof of the theorem, which proves Proposition~\ref{prop:lowersemicontinuity}.

\begin{proof}[Proof of Theorem~\ref{thm:openmap}]
Let  $u^*\in \mathcal{U}$ and $\{\succeq_n\}$ be a sequence in $\R$
with  $\succeq_n \rightarrow \Phi(u^*)$. By Lemma~\ref{lem:openmap1}
it is enough to exhibit a sequence $u_n\in\Phi^{-1}(\succeq_n)$ and
$u_n \rightarrow u^*$ in the topology of compact convergence.

Let $u_n = u_{\succeq_n}$ as defined in Lemma~\ref{lem:openmap2} from
$u^*$. Lemma~\ref{lem:openmap2} implies that
$u_n\in\Phi^{-1}(\succeq_n)$.
By XII Theorem 7.5 p. 268 of \citet{Dugundji}, to establish compact
convergence it is enough to show that for any convergent sequence
$\{x_n\}$, with $x_n\rightarrow x^*$, $u_n(x_n)\rightarrow u^*(x^*)$.

To this end, let $x_n\rightarrow x^*$. Let $\hat m = m^*(x^*)$ and
$m_n \sim_n x_n$, using the notation in Lemma~\ref{lem:openmap2}, and  $U$ be a neighborhood
of $\hat m$. Let $\ul m,\overline m\in M$ be such that $\ul m < \hat m <
\overline m$ and $ [\ul m,\overline m]\subseteq U$. Then it must be true that $m^n\in [\ul m,\overline
m]$ for all $n$ large enough. To see this, note that if, for example,
$m^n\geq \overline m$ infinitely often then there would exist a subsequence
for which $x^n \succeq^n m^n\succeq \overline m$ (by monotonicity of
$\succeq$), which would imply that $x^*\succeq \overline m > \hat m$, as
$\succeq_n\rightarrow \succeq$. But $\hat m \sim x^*\succeq \overline m$ is
a  violation of monotonicity.

Now $m^n\in [\ul m, \overline m]\subseteq U$ for all $n$ large enough means
that $m^n\rightarrow \hat m$. Thus \[u_{\succeq_n} (x_n) = u^* (x_n)
\rightarrow u^*(x^*) = u_{\succeq} (x^*),\] as $u^*$ is continuous.
\end{proof}

\section{Proof of Theorem~\ref{thm:expectedutility}}



By Theorem 8 of \citet{border1994dynamic}, $\Phi(\calv)$ is compact, and therefore $\succeq_k$ possesses a limit point $\succeq^*\in\Phi(\calv)$.  By Lemma~\ref{lem:inclusion}, the set of $\succeq_k$ weakly rationalizing $c^k$ is closed, and hence compact.  Suppose by means of contradiction that there is some $\succeq'_k$ also weakly rationalizing $c^k$ which converges to $\succeq\neq\succeq^*$.  Observe that each of $\succeq^*$ and $\succeq$ weakly rationalize each $c^k$.

Finally, let $\succeq'$ be any complete relation such that for all $(a,b)\in B\times B$, $a \succeq' b$ if and only if $a\in c^k(\{a,b\})$, for $k$ such that $\{a,b\}\in\Sigma^k$.  Then, by definition of weak rationalization, we have $\succeq'_{B\times B}\subseteq \succeq_{B\times B}\cap \succeq^*_{B\times B}$.  Appeal then to Lemma~\ref{lem:partialobservability} to conclude that $\succeq=\succeq^*$, a contradiction.



\end{document}